\documentclass{amsart}
\usepackage{amssymb}
\newtheorem{theorem}[subsection]{Theorem}
\newtheorem{lemma}[subsection]{Lemma}
\newtheorem{proposition}[subsection]{Proposition}
\newtheorem{corollary}[subsection]{Corollary}

\theoremstyle{definition}

\newtheorem{definition}[subsection]{Definition}

\theoremstyle{remark}

\newtheorem{remark}[subsection]{Remark}

\newtheorem{example}[subsection]{Example}

\newcommand{\Id}{\operatorname{Id}}

\newcommand{\GL}{\operatorname{GL}}

\vspace{3.5cm}
\title[An integrable deformation of a Lie algebra of skew hermitian 
$\mathbb{Z} \times \mathbb{Z} $-matrices]{A construction of solutions of an integrable deformation of a commutative Lie algebra of skew hermitian 
$\mathbb{Z} \times \mathbb{Z} $-matrices}

\author{Aloysius G. Helminck}
\address{Department of Mathematics,
University of Hawaii at Manoa,
Honolulu, HI 96822}
\email{helminck@hawaii.edu}
\author{Gerardus F.  Helminck}
\address{KdV Institute,
University of Amsterdam,
1090 GE Amsterdam,
The Netherlands}
\email{g.f.helminck@uva.nl}
\thanks{Dedicated to the memory of G. van Dijk}

\begin{document}

\begin{abstract}
Inside the algebra $LT_{\mathbb{Z}}(R)$ of $\mathbb{Z} \times \mathbb{Z}$-matrices with coefficients from a commutative $\mathbb{C}$-algebra $R$
that have only a finite number of nonzero diagonals above the central diagonal, we consider a deformation of a commutative Lie algebra $\mathcal{C}_{sh}(\mathbb{C})$ of finite band skew hermitian matrices that is 
different from the Lie subalgebras that were deformed at the discrete KP hierarchy and its strict version.
The evolution equations that the deformed generators of $\mathcal{C}_{sh}(\mathbb{C})$ 
have to satisfy are determined by the decomposition of $LT_{\mathbb{Z}}(R)$ in the direct sum of an algebra of lower triangular matrices and the finite band skew hermitian matrices. This yields then 
the $\mathcal{C}_{sh}(\mathbb{C})$-hierarchy. 
We show that the projections of a solution satisfy zero curvature relations and that it suffices to solve an associated Cauchy problem. Solutions of this type can be  obtained by finding appropriate vectors in the $LT_{\mathbb{Z}}(R)$-module of oscillating matrices, the so-called wave matrices, that satisfy a set of equations in the oscillating matrices, 
called the linearization of the $\mathcal{C}_{sh}(\mathbb{C})$-hierarchy. Finally, a Hilbert Lie group will be introduced from which wave matrices for the $\mathcal{C}_{sh}(\mathbb{C})$-hierarchy are constructed. There is a real analogue of the $\mathcal{C}_{sh}(\mathbb{C})$-hierarchy called the $\mathcal{C}_{as}(\mathbb{R})$-hierarchy. It consists of a deformation of a commutative Lie algebra $\mathcal{C}_{as}(\mathbb{R})$ of anti-symmetric matrices. We will properly introduce it here too on the way and mention everywhere the corresponding result for this hierarchy, but we leave its proofs mostly to the reader.
\end{abstract}

\maketitle



\section{\boldmath Introduction}

The major part of the scientific work of G. van Dijk centered around developing representation theoretic questions
connected with vector bundles over homogeneous spaces $G/H$, where $G$ is a locally compact group and $H$ a closed subgroup of $G$. A main goal in this field is to decompose a given representation in components that cannot be split anymore, the so-called irreducible components. 
This leads to a wide range of research problems to which he contributed significantly. We mention a few:
a continuous search for new ways and techniques to decompose, what is the role of the structure of the space $G/H$ in the decomposition and which of the known representations fit in the new one. The following selection illustrates the width of his interests: 
\cite{vD74}, \cite{vD78}, \cite{vD86}, \cite{vD+AP1999}, \cite{vD+YS2000}, \cite{vD+MP01} and \cite{vD2009}.
Gerrit mostly worked with concrete spaces and favorite examples of the spaces $G/H$ were the symmetric spaces, see \cite{vD+MK86}, \cite{vD+MP86}, \cite{vD+EB94}, \cite{vD+SH97} and \cite{vD+VM03}. Defined in the most general way, symmetric spaces are spaces $G/H$, with $G$ a topological group and $H$ the fixed point group of a continuous involution $\sigma$ of $G$. 

In the present paper we present another role of symmetric spaces, namely as a mean to construct solutions for 
compatible systems of Lax equations, so-called integrable hierarchies, for difference operators or pseudo difference operators. A wide collection of compatible systems have been discussed, both on the physical, see \cite{OT1},\cite{OT2} and \cite{OT3}, as on the mathematical side, see e.g. \cite{UT}, \cite{Takebe}, \cite{APvM99} and \cite{HPP1}. Here we discuss a decomposition of 
$LT_{\mathbb{Z}}(R)$ into the direct sum of two Lie algebras different from the ones used in \cite{APvM99} and \cite{HPP1} to get the discrete KP hierarchy and its strict version and deform a commutative Lie algebra $\mathcal{C}_{sh}(\mathbb{C})$ of skew hermitian matrices 
leading to the $\mathcal{C}_{sh}(\mathbb{C})$-hierarchy. 
The projections of a solution of the hierarchy are shown to satisfy zero curvature relations and we present a Cauchy problem whose solutions are sufficient to produce solutions of the hierarchy. Solutions of this Cauchy problem can be obtained by finding appropriate vectors in the $LT_{\mathbb{Z}}(R)$-module of oscillating matrices. These vectors are called wave matrices and satisfy a set of equations in the oscillating matrices, called the linearization of the $\mathcal{C}_{sh}(\mathbb{C})$-hierarchy. Finally, a Hilbert Lie group $G(2)(\mathbb{C})$ will be introduced from which wave matrices for the $\mathcal{C}_{sh}(\mathbb{C})$-hierarchy are constructed.

The contents of the various sections is as follows: Section \ref{sec2} describes the algebra $LT_{\mathbb{Z}}(R)$, its relevant decomposition and the commutative algebra $\mathcal{C}_{sh}(\mathbb{C})$. 
Further we 
present there the type of deformation that will be considered in $LT_{\mathbb{Z}}(R)$. 
In section \ref{sec3} 
we cluster a number of Lax equations that the deformation has to satisfy, into an integrable hierarchy
and show that the relevant projections of the deformed basic 
directions satisfy zero curvature relations. Further we present a Cauchy problem, whose solutions lead directly to solutions of the $\mathcal{C}_{sh}(\mathbb{C})$-hierarchy. A good context to produce these solutions is the $LT_{\mathbb{Z}}(R)$-module of oscillating matrices. Appropriate vectors in this module, the so-called wave matrices, satisfy a set of equations in the oscillating matrices, the linearization of the $\mathcal{C}_{sh}(\mathbb{C})$-hierarchy, that leads to solutions of the hierarchy.  Finally, one introduces a Hilbert Lie group $G(2)(\mathbb{C})$ and its unitary subgroup  $U(2)$ and constructs for each coset $gU(2), g \in G(2)(\mathbb{C}),$ a wave matrix of the $\mathcal{C}_{sh}(\mathbb{C})$-hierarchy.

\section{The algebra $LT_{\mathbb{Z}}(R)$}
\label{sec2}

The algebra where the central deformation of this paper takes place is that of the complex pseudo difference operators Ps$\Delta$. We use its realization $LT_{\mathbb{Z}}(R)$ as a subset of the space $M_{\mathbb{Z}}(R)$ of $\mathbb{Z} \times \mathbb{Z}$-matrices with 
coefficients 
from a commutative $\mathbb{C}$-algebra $R$. The algebras $R$ we work with throughout this paper are the complexifications of an algebra of real-valued functions $R(\mathbb{R})$, i.e. $R=\mathbb{C} \otimes_{\mathbb{R}} R(\mathbb{R})$ and we will denote $\alpha \otimes f$, $\alpha \in \mathbb{C}$ and $f \in R(\mathbb{R})$, simply by $\alpha f$. On $R$ complex conjugation is defined by $\overline{\alpha f}:=\overline{\alpha}f$. An element $f$ from $R(\mathbb{R})$ is called positive, if all its values are, and we write then $f>0$.
We start by recalling a number of basic notations in the algebra $LT_{\mathbb{Z}}(R)$

Each $A \in M_{\mathbb{Z}}(R)$ will be denoted as $A=(a_{ij})$ or as $A=(a_{(i,j)})$ if confusion in the labeling might occur.
On the space $M_{\mathbb{Z}}(R)$ we use the ordering of
columns and rows as in the finite dimensional case. 
The transpose $A^{T}$ of a matrix $A \in M_{\mathbb{Z}}(R)$ is given by the matrix $(a_{ji})$ and  the adjoint $A^{*}$ of $A$ is the matrix $(\overline{a_{ji}})$.
Any $A \in M_{\mathbb{Z}}(R)$ corresponds to an $R$-linear map.
Consider thereto the space of all $1  \times \mathbb{Z}$-matrices with coefficients from $R$
$$V=R^{\mathbb{Z}}=\{\vec{x}=(x_{n})=\left(
\begin{matrix}
&\hdots
&x_{n-1} 
&x_{n}
&x_{n+1}
&\hdots
\end{matrix}
\right) \mid x_{n} \in R\}
$$ 
and its 
subspace 
$$
V_{{\rm fin}}=\{ \vec{x}=(x_{n}) \in V \mid  x_{n} \neq 0 \text{ for only a finite number of }n \}.
$$
Define for each $i \in \mathbb{Z}$ the vector $\vec{e}\,(i) $ in $V_{{\rm fin}}$ by requiring its $i$-th coordinate to be equal to one and its remaining coordinates to be zero. 
Then $V_{{\rm fin}}$ is a free $R$-module with basis the $\{\vec{e}\,(i) \mid i \in \mathbb{Z} \}$. 
On $V_{{\rm fin}}$
we can define an $R$-linear action $M_{A}$ of $A=(a_{nm})$ by 
\begin{equation}
\label{algmult}
M_{A}(\vec{x}):= \vec{x}A .
\end{equation}
Hence, the matrix $A$ determines the $R$-linear 
map $M_{A} \in {\rm Hom}_{R}(V_{{\rm fin}},V)$. 

Next we present 
two types of matrices $A$ that generate the algebra Ps$\Delta$ and for which
$M_{A}$ is even defined on $V$. 
The first class 
is that of the diagonal matrices.
Given a collection of elements
$\{d(s) \mid  s
\in
\mathbb{Z}
\}$ from $R$, this defines 
the diagonal matrix
diag($d(s)$) in $M_{\mathbb{Z}}(R)$ with $d(s)$ as its ($s,s$)-entry.
The algebra of all diagonal matrices in $M_{\mathbb{Z}}(R)$ is denoted by $\mathcal{D}_{1}(R)$
and its group of units by $\mathcal{D}_{1}(R)^{*}$, i.e. all $\text{diag}(d(s)) \text{ with } d(s) \in R^{*} \text{ for all } s \in  \mathbb{Z}.$ We get an embedding $j_{1}: R \to \mathcal{D}_{1}(R)$ by putting $j_{1}(r)=r \Id$ for all $r \in R$.

The second class of examples form the shift matrix $\Lambda$, its inverse $\Lambda^{-1}$
and their powers, where the first corresponds to $M_{\Lambda}(\vec{e}\,(i))=\vec{e}\,(i+1)$.
The group $\{ \Lambda^{m} \mid m \in \mathbb{Z} \}$ normalizes $\mathcal{D}_{1}(R)$, for there holds for all $d \in \mathcal{D}_{1}(R)$
\begin{equation}
\label{Ldiag}
\Lambda^{m}  d \Lambda^{-m}=\Lambda^{m}  \text{diag}(d(s)) \Lambda^{-m}=\text{diag}(d(s+m)).
\end{equation}
A convenient tool in $M_{\mathbb{Z}}(R)$ is decomposing a matrix $A=(a_{ij}) \in M_{\mathbb{Z}}(R)$ in its diagonals. 
If $m \in \mathbb{Z}$, then the $m$-th {\it diagonal} of $A$ is by definition 
$$
d_{m}(A) \Lambda^{m}, \text{ with } d_{m}(A)=\text{diag}(a_{(s,s+m)}) \in \mathcal{D}_{1}(R),
$$
and $m$ determines if the diagonal is positive or negative. 
Then each matrix can be split as
 \begin{equation}
\label{diagdeco}
 A=\sum_{m \in \mathbb{Z}} d_{m}(A)\Lambda^{m}
\end{equation}
Let $LT_{\mathbb{Z}}(R)$ be the collection of all matrices in $M_{\mathbb{Z}}(R)$ that have only a finite number of nonzero positive diagonals. Relation (\ref{Ldiag}) implies now the following property
\begin{lemma}
\label{L1.1}
If $A \in LT_{\mathbb{Z}}(R)$ is equal to its $\ell$-th diagonal and $B \in LT_{\mathbb{Z}}(R)$ is equal to its $n$-the diagonal, then $AB$ is equal to its $\ell +n$-th diagonal. In particular, $LT_{\mathbb{Z}}(R)$ is an algebra w.r.t. matrix multiplication.
\end{lemma}
 We use the decomposition in (\ref{diagdeco}) to assign a degree to elements of $LT_{\mathbb{Z}}(R)$. For a nonzero $A $ in $LT_{\mathbb{Z}}(R)$ the degree is equal to $m$ if its highest nonzero diagonal is the $m$-th and the degree of the zero element is $- \infty$.

The algebra $LT_{\mathbb{Z}}(R)$ possesses a large collection of invertible elements. For, let 
$V \in LT_{\mathbb{Z}}(R)$ have the form
$
V=\sum_{i \leqslant m} v_i \Lambda^{i},
$
with 
$v_m \in  \mathcal{D}_1(R)^{*}$. Then one shows recursively
\begin{lemma}
\label{L2.1}
Each element $V$ in $LT_{\mathbb{Z}}(R)$ 
with an invertible leading coefficient 
is invertible in $LT_{\mathbb{Z}}(R)$. 
This class of invertible elements in $LT_{\mathbb{Z}}(R)$ 
forms the group $I(LT_{\mathbb{Z}}(R))$. 
\end{lemma}

Next we discuss the relevant decompositions of $LT_{\mathbb{Z}}(R)$. 
In the complex case we consider the real Lie subalgebra $\mathcal{SH}(R)$ of skew hermitian matrices in $LT_{\mathbb{Z}}(R)$, i.e.
$$
\mathcal{SH}(R)=\{ A \mid A \in LT_{\mathbb{Z}}(R), A^{*}=-A \}.
$$
A general element  $A$ in $\mathcal{SH}(R)$ has the form
\begin{equation}
\label{formSH}
A=\sum_{j=1}^{N} d_{j}(A)\Lambda^{j} +d_{0}(A)- \sum_{j=1}^{N} \Lambda^{-j}d_{j}(A)^{*}, 
\end{equation}
with $d_{0}(A) \in i \mathcal{D}_{1}(R(\mathbb{R}))$$\text{ and all remaining }d_{j}(A) \in \mathcal{D}_{1}(R)$.
Inside  
$\mathcal{SH}(\mathbb{C})$ we consider the real Lie subalgebra $\mathcal{C}_{sh}(\mathbb{C})$ spanned by the elements
\begin{equation}
\label{elCsh}
G_{j1}=\Lambda^{j}-\Lambda^{-j}, j \geqslant 1, G_{02}=i \Id, G_{j2}=i(\Lambda^{j}+\Lambda^{-j}), j \geqslant 1.
\end{equation}
It is convenient to introduce a notation for the index set of the basis (\ref{elCsh}) of $\mathcal{C}_{sh}(\mathbb{C})$. We write $\Sigma_{1}=\{ j1 \mid j \geqslant 1 \}$, $\Sigma_{2}=\{ j2 \mid j \geqslant 0 \}$ and $\Sigma=\Sigma_{1} \cup \Sigma_{2}$.
The Lie algebra $\mathcal{C}_{sh}(\mathbb{C})$ is clearly commutative. Moreover, it is maximal in 
$\mathcal{SH}(\mathbb{C})$ with respect to this property. For, let $A$ in 
$\mathcal{SH}(\mathbb{C})$ 
be an element that commutes with all the $\{ G_{\sigma} \mid \sigma \in \Sigma \}$. The matrix $A$ has the form (\ref{formSH}), where all diagonal matrices $\{ d_{j}(A) \mid j \geqslant 0\}$ are written as $c_{j}+i d_{j}$ with $c_{j}$ and $d_{j} \in \mathcal{D}_{1}(\mathbb{R})$ and $c_{0}=0$.
Because of Lemma \ref{L1.1}, the fact that $A$ and $G_{11}$ commute implies that their leading terms commute. 
If the leading term in $\Lambda $ is $d_{0}(A)$, then $d_{0}(A)=i j_{1}(s_{0}), s_{0} \in \mathbb{R},$ and $A=s_{0}G_{02} \in \mathcal{C}_{sh}(\mathbb{C})$. If the leading term in $\Lambda $ is $d_{N}(A)\Lambda^{N}, N \geqslant 1$, then this implies that $d_{N}(A)$ has the form $d_{N}(A)=j_{1}(r_{N})+i j_{1}(s_{N})$ with $r_{N} \text{ and }s_{N} \in \mathbb{R}$. Consider now the element $A-r_{N}G_{N1}-s_{N}G_{N2}$. It still commutes with all the $\{ G_{\sigma} \mid \sigma \in \Sigma \}$ and has a leading term in $\Lambda $ of degree lower than $N$. Thus we have shown by induction with respect to $N$ 
\begin{lemma}
\label{L2.2}
The Lie algebra $\mathcal{C}_{sh}(\mathbb{C})$ is a maximal commutative Lie subalgebra of $\mathcal{SH}(\mathbb{C})$.
\end{lemma}
So, $\mathcal{C}_{sh}(\mathbb{C})$ is optimal as for commutativity. 
We call the basis $\{ G_{\sigma} \mid \sigma \in \Sigma \}$ of $\mathcal{C}_{sh}(\mathbb{C})$ the {\it basic directions} of this space. From the multiplication rules for diagonals and the general form of the elements of $\mathcal{SH}(R)$, one sees that the space
\begin{equation*}
P_{-}(\mathbb{R})=\{P=\sum_{j \leqslant 0} d_{j}(P) \Lambda^{j} \mid d_{0}(P) \in \mathcal{D}_{1}(R(\mathbb{R})),  d_{j}(P)  \in \mathcal{D}_{1}(R) \text{ for } j<0 \}
\end{equation*}
is a real Lie subalgebra of $LT_{\mathbb{Z}}(R)$ that complements $\mathcal{SH}(R)$, i.e.
\begin{equation}
\label{Cdirsum}
LT_{\mathbb{Z}}(R)=P_{-}(\mathbb{R}) \oplus \mathcal{SH}(R).
\end{equation}
Let $P=\sum_{j \leqslant N}d_{j}(P)\Lambda^{j}$ be a general element of $LT_{\mathbb{Z}}(R)$. Decompose each  $d_{j}(P)$ as $d_{j}(P)=a_{j}+ib_{j}$ with $a_{j} \text{ and }b_{j} \in \mathcal{D}_{1}(R(\mathbb{R}))$. Then the projection $\pi_{-}$ from $LT_{\mathbb{Z}}(R)$ onto $ P_{-}(\mathbb{R})$ in the decomposition (\ref{Cdirsum}) is given by
\begin{equation*}
\pi_{-}(P)=a_{0}+ \sum_{j<0}d_{j}(P)\Lambda^{j} +\sum_{j<0} \Lambda^{j} d_{-j}(P)^{*}
\end{equation*}
and the projection $\pi_{sh}$  from $LT_{\mathbb{Z}}(R)$ onto the second component in 
the decomposition (\ref{Cdirsum}) is given by
\begin{equation*}
\pi_{sh}(P)=P-\pi_{-}(P)=ib_{0}+ \sum_{j>0}d_{j}(P)\Lambda^{j} -\sum_{j>0} \Lambda^{-j} d_{j}(P)^{*}
\end{equation*}
Next we assign a group to the Lie algebra $ P_{-}(\mathbb{R})$. For, if the exponential map is well defined on $P_{-}(\mathbb{R})$, then the image under $\exp$ of a $P=\sum_{j \leqslant 0}d_{j}(P)\Lambda^{j} \in P_{-}(\mathbb{R})$ is a lower triangular matrix with leading term $\exp (d_{0}(P))$, which is invertible in $\mathcal{D}_{1}(R(\mathbb{R}))$ with inverse $\exp (-d_{0}(P))$ and moreover $\exp (d_{0}(P))>0$, i.e. , if $\exp (d_{0}(P))=$diag($d(s)$), then $d(s) >0$ for all $s \in \mathbb{Z}$. In particular the image of $\exp$ belongs to the group inside $LT_{\mathbb{Z}}(R)$ given by
\begin{equation*}
\mathcal{P}_{-}(\mathbb{R})=\{G=\sum_{j \leqslant 0} d_{j}(G) \Lambda^{j} \in P_{-}(\mathbb{R}) \mid d_{0}(G) \in \mathcal{D}_{1}(R(\mathbb{R}))^{*} \text{ and } d_{0}(G) >0\}
\end{equation*}
Therefore we see $\mathcal{P}_{-}(\mathbb{R})$ as the group associated with the Lie algebra $P_{-}(\mathbb{R})$. 
In the sequel we will look at deformations of $\mathcal{C}_{sh}(\mathbb{C})$ by conjugating with elements from $\mathcal{P}_{-}(\mathbb{R})$.  Therefore we introduce now the following notion:

\begin{definition}
\label{D2.1}
A $\mathcal{P}_{-}(\mathbb{R})$-deformation of the $\{ G_{\sigma} \mid \sigma \in \Sigma \}$ is a collection of matrices $\{ \mathcal{G}_{\sigma} \mid \sigma \in \Sigma \}$ 
in $LT_{\mathbb{Z}}(R)$ such that for all $\sigma \in \Sigma$,
$
\mathcal{G}_{\sigma}=gG_{\sigma}g^{-1},
$
for some $g \in \mathcal{P}_{-}(\mathbb{R})$. We call $g$ the {\it dressing operator} of the deformation.
\end{definition}
Since the element $G_{02}$ remains the same at this type of deformations, it suffices to focus on the deformation of the remaining basic directions $\{ G_{\sigma} \mid \sigma \in \Sigma_{0} \}$ of $\mathcal{C}_{sh}(\mathbb{C})$ , where $\Sigma_{0}=\Sigma_{1} \cup  \{ \sigma=j2 \mid j \geqslant 1\}$.

\begin{remark}
\label{R2.0}
There is a real analogue of the deformation described above. It takes place in $LT_{\mathbb{Z}}(R_{as})$, where $R_{as}$ is a commutative algebra of real-valued functions.
Now we consider the real Lie subalgebra $\mathcal{AS}(R_{as})$ of antisymmetric matrices in $LT_{\mathbb{Z}}(R_{as})$, i.e.
$$
\mathcal{AS}(R_{as})=\{ A \mid A \in LT_{\mathbb{Z}}(R_{as}), A^{T}=-A \}.
$$
A general element  $A$ in $\mathcal{AS}(R_{as})$ has the form
\begin{equation}
\label{formAS}
A=\sum_{j=1}^{N} d_{j}(A)\Lambda^{j} - \sum_{j=1}^{N} \Lambda^{-j}d_{j}(A), 
\end{equation}
$\text{with all }d_{j}(A) \in \mathcal{D}_{1}(R_{as})$.
Inside  
$\mathcal{AS}(\mathbb{R})$ we consider the real Lie subalgebra $\mathcal{C}_{as}(\mathbb{R})$ spanned by the elements
\begin{equation}
\label{elCas}
F_{j}=\Lambda^{j}-\Lambda^{-j}, j \geqslant 1.
\end{equation}
We write $\mathcal{J}=\{j \in \mathbb{N} \mid j \geqslant 1 \}$.
Similarly to $\mathcal{C}_{sh}(\mathbb{C})$, the Lie algebra $\mathcal{C}_{as}(\mathbb{R})$ is a maximal commutative Lie subalgebra of $\mathcal{AS}(\mathbb{R})$.
We call the basis $\{ F_{j} \mid j \in \mathcal{J} \}$ also the {\it basic directions}  of $\mathcal{C}_{as}(\mathbb{R})$. From (\ref{formAS})
one sees directly that the space of lower triangular matrices in $LT_{\mathbb{Z}}(R_{as})$
\begin{equation*}
P_{\leqslant 0}=\{P=\sum_{j \leqslant 0} d_{j}(P) \Lambda^{j} \mid 
d_{j}(P)  \in \mathcal{D}_{1}(R_{as}) \text{ for all } j \leqslant 0 \}
\end{equation*}
is a real Lie subalgebra of $LT_{\mathbb{Z}}(R_{as})$ that complements $\mathcal{AS}(R_{as})$, i.e.
\begin{equation}
\label{Casdirsum}
LT_{\mathbb{Z}}(R_{as})=P_{\leqslant 0} \oplus \mathcal{AS}(R_{as}).
\end{equation}
Let $P=\sum_{j \leqslant N}d_{j}(P)\Lambda^{j}$ be a general element of $LT_{\mathbb{Z}}(R_{as})$. 
Then the projection $\pi_{\leqslant 0}$ from $LT_{\mathbb{Z}}(R_{as})$ onto $ P_{\leqslant 0}$ in the decomposition (\ref{Casdirsum}) is given by
\begin{equation*}
\pi_{\leqslant 0}(P)= \sum_{j \leqslant 0}d_{j}(P)\Lambda^{j} +\sum_{j<0} \Lambda^{j} d_{-j}(P)
\end{equation*}
and the projection $\pi_{as}$  from $LT_{\mathbb{Z}}(R)$ onto the second component in 
the decomposition (\ref{Casdirsum}) is given by
\begin{equation*}
\pi_{as}(P)=P-\pi_{\leqslant 0}(P)= \sum_{j>0}d_{j}(P)\Lambda^{j} -\sum_{j>0} \Lambda^{-j} d_{j}(P)
\end{equation*}
Inside $LT_{\mathbb{Z}}(R_{as})$ we have the group 
\begin{equation*}
\mathcal{P}_{\leqslant 0}=\{P=\sum_{j \leqslant 0} d_{j}(P) \Lambda^{j} \in P_{-} \mid d_{0}(P) \in \mathcal{D}_{1}(R_{as})^{*} \text{ and } d_{0}(P) >0\}
\end{equation*}
that we see as the group associated with the Lie algebra $P_{\leqslant 0}$ for a similar reason as for $\mathcal{P}_{-}(\mathbb{R})$.  
Our interest is again in deformations of $\mathcal{C}_{as}(\mathbb{R})$ by conjugating with elements from $\mathcal{P}_{\leqslant 0}$ and we call a set of matrices $\{ \mathcal{F}_{j} \mid j \in \mathcal{J} \}$ 
in $LT_{\mathbb{Z}}(R)$ a $\mathcal{P}_{\leqslant 0}$-{\it deformation} of the $\{ F_{j} \mid j \in \mathcal{J} \}$, if there is a $g \in \mathcal{P}_{\leqslant 0}$, such that for all $j \in \mathcal{J} $,
$
\mathcal{F}_{j}=gF_{j}g^{-1}.
$
The element $g$ is called again the {\it dressing operator} of the deformation.
\end{remark}

\section{The integrable hierarchies}
\label{sec3}

In this section we discuss first the evolution equations that we require of a $\mathcal{P}_{-}(\mathbb{R})$-deformation of the basis $\{ G_{\sigma} \mid \sigma \in \Sigma \}$ of $\mathcal{C}_{sh}(\mathbb{C})$.
 We need then that $R(\mathbb{R})$ is equipped with a set of commuting $\mathbb{R}$-linear derivations $\{ \partial_{\sigma} \mid \sigma \in \Sigma \}$. Each $\partial_{\sigma}$ is an algebraic substitute for differentiation with respect to the flow parameter of the flow corresponding to $G_{\sigma}$. We extend all $\{\partial_{\sigma} \}$ to $\mathbb{C}$-linear derivations of $R$ by putting $\partial_{\sigma}(\alpha f):=\alpha \partial_{\sigma}(f)$ for all $\alpha \in \mathbb{C} $ and $f \in R(\mathbb{R})$. By letting each $\partial_{\sigma}$ act on the matrix coefficients of an element of $LT_{\mathbb{Z}}(R)$ one gets a $\mathbb{C}$-linear derivation of this algebra which is denoted by the same symbol. The evolution equations we require of a $\mathcal{P}_{-}(\mathbb{R})$-deformation $\{ \mathcal{G}_{\sigma} \mid \sigma \in \Sigma \}$ of the basic directions of $\mathcal{C}_{sh}(\mathbb{C})$, are determined by the decomposition (\ref{Cdirsum}) and consist of all equations
\begin{equation}
\label{Laxsh}
\partial_{\sigma_{1}}(\mathcal{G}_{\sigma_{2}} )=[\pi_{sh}(\mathcal{G}_{\sigma_{1}} ),\mathcal{G}_{\sigma_{2}} ]=[\mathcal{G}_{\sigma_{2}} ,\pi_{-}(\mathcal{G}_{\sigma_{1}} )]
\end{equation}
for all $\sigma_{1}$ and $\sigma_{2}$ from $\Sigma$. The second equality in (\ref{Laxsh}) is a direct consequence of the fact that all $\mathcal{G}_{\sigma_{1}}$ and $\mathcal{G}_{\sigma_{2}}$ commute. The data $(R, \{ \partial_{\sigma}\})$ we call a {\it setting} in which we can consider these deformations and their evolution equations. A $\mathcal{P}_{-}(\mathbb{R})$-deformation of the $\{ G_{\sigma} \mid \sigma \in \Sigma \}$ that satisfies all the equations (\ref{Laxsh}), is called a {\it solution} in the setting $(R, \{\partial_{\sigma}\})$ of the $\mathcal{C}_{sh}(\mathbb{C})$-hierarchy after the commutative Lie algebra that gets deformed.
The equations (\ref{Laxsh}) are called the {\it Lax equations} of the $\mathcal{C}_{sh}(\mathbb{C})$-hierarchy. These equations always have the trivial solution  $\{ \mathcal{G}_{\sigma}=G_{\sigma} \mid \sigma \in \Sigma \}$. Since at each $\mathcal{P}_{-}(\mathbb{R})$-deformation $\{ \mathcal{G}_{\sigma}\}$ there holds $\mathcal{G}_{02}=G_{02}=i\Id$, the derivation $\partial_{02}$ is zero on all $\{ \mathcal{G}_{\sigma}\}$ and it suffices to prove the equations  (\ref{Laxsh})
for all $\sigma_{1}$ and $\sigma_{2}$ from $\Sigma_{0}=\{\sigma \in \Sigma \mid \sigma \neq 02 \}$.

\begin{remark}
\label{R3.0} 
For the $\mathcal{P}_{\leqslant 0}$-deformation of the basic directions of $\mathcal{C}_{as}(\mathbb{R})$ from Remark (\ref{R2.0}) there is also a natural set of evolution equations that one can consider.
We need now that $R_{as}$ is equipped with a set of commuting $\mathbb{R}$-linear derivations $\{ \partial_{j} \mid j \in \mathcal{J} \}$. Each $\partial_{j}$ is an algebraic substitute for differentiation with respect to the flow parameter of the flow corresponding to $F_{j}$. 
By letting each $\partial_{j}$ act on the matrix coefficients of an element of $LT_{\mathbb{Z}}(R_{as})$ one gets a $\mathbb{R}$-linear derivation of this algebra which is also denoted by $\partial_{j}$. The evolution equations we require of a $\mathcal{P}_{\leqslant 0}$-deformation $\{ \mathcal{F}_{\sigma} \mid \sigma \in \Sigma \}$ of the basic directions of $\mathcal{C}_{as}(\mathbb{R})$, are determined by the decomposition (\ref{Casdirsum}) and consist of all equations
\begin{equation}
\label{Laxas}
\partial_{j_{1}}(\mathcal{F}_{j_{2}} )=[\pi_{as}(\mathcal{F}_{j_{1}} ),\mathcal{F}_{j_{2}} ]=[\mathcal{F}_{j_{2}} ,\pi_{\leqslant 0}(\mathcal{F}_{j_{1}} )]
\end{equation}
for all $j_{1}$ and $j_{2}$ from $\mathcal{J}$. The second equality in (\ref{Laxas}) follows again directly from
the fact that all $\mathcal{F}_{j_{1}}$ and $\mathcal{F}_{j_{2}}$ commute.  Also the data $(R_{as}, \{ \partial_{j}\})$ we call a {\it setting} in which we can consider the present deformations and their evolution equations. 
A $\mathcal{P}_{\leqslant 0}$-deformation of the $\{ F_{j} \mid j \in \mathcal{J} \}$ that satisfies all the equations (\ref{Laxas}), is called a {\it solution} in the setting $(R_{as}, \{\partial_{j}\})$ of the $\mathcal{C}_{as}(\mathbb{R})$-hierarchy after the commutative Lie algebra that gets deformed.
The equations (\ref{Laxas}) are called the {\it Lax equations} of the $\mathcal{C}_{as}(\mathbb{R})$-hierarchy. Also these equations always have a trivial solution  $\{ \mathcal{F}_{j}=F_{j} \mid j \in \mathcal{J} \}$.
\end{remark}
\begin{example}
\label{E3.2}
To give an idea of what to expect of Lax equations like (\ref{Laxsh}) and (\ref{Laxas}) we present here an example of a simple system that can be put in this form. Recall
that the so-called {\it infinite Toda chain}  consists of 
an infinite number of
particles on a straight line labeled by $\mathbb{Z}$. We assume for simplicity that they all have the same mass equal to one
and that their equations of motion 
are given by  
\begin{equation}
\label{1.1}
\frac{dq_n}{dt}=p_n \;\;\mbox{and}
\;\; \frac{dp_n}{dt}=2e^{2(q_{n-1}-q_{n})}-2e^{2(q_{n}-q_{n+1})},\;\;n  \in
\mathbb{Z}.
\end{equation}
Here $q_n$ is the displacement of the $n$-th particle, $p_{n}$  its momentum
and the two exponential factors in equation (\ref{1.1}) describe the forces exerted on the $n$-th particle by each of its neighbors. 
These equations can be
rewritten as an equality between $\mathbb{Z} \times \mathbb{Z}$-matrices.
Thereto we put
\begin{equation*}\notag
a_n:=
e^{q_n-q_{n+1}} .
\end{equation*}
The equations (\ref{1.1}) get then the following form
\begin{equation}\label{1.2}
\frac{da_n}{dt}=a_n(p_n-p_{n+1}) 
\;\;\mbox{and}\;\;\frac{dp_n}{dt}=2(a_{n-1}^2-a_n^2),\;\;n \in \mathbb{Z}.
\end{equation}
Consider now the 
$\mathbb{Z} \times \mathbb{Z}$-matrices $L$ and $M$ of the form
\begin{equation*}\notag
L= 
\left( 
\begin{matrix}
 \ddots & \ddots & \ddots & &0\\
\ddots & p_{n-1}&a_{n-1}&0& \ddots \\
\ddots &a_{n-1}& p_{n}&a_{n}& \ddots\\
&0&a_{n}& p_{n+1}& \ddots\\
0& &\ddots & \ddots & \ddots
\end{matrix}
\right)
\;\;\mbox{and}\;\;
M=\left( 
\begin{matrix}
 \ddots & \ddots & \ddots & &0\\
\ddots & 0&a_{n-1}&0& \ddots \\
\ddots &-a_{n-1}& 0&a_{n}& \ddots\\
&0&-a_{n}& 0& \ddots\\
0& &\ddots & \ddots &  \ddots
\end{matrix}
\right).
\end{equation*}
The matrices $L$ and $M$ decompose as follows in diagonal matrices and powers of $\Lambda$:
\begin{align}
\label{decoL1}
L&={\rm diag}(a_{n}) \Lambda +{\rm diag}(p_{n}) \Id +\Lambda^{-1} {\rm diag}(a_{n}) \\ \notag
&=M +{\rm diag}(p_{n}) \Id +\Lambda^{-1} {\rm diag}(2a_{n}),
\end{align}
where $M={\rm diag}(a_{n}) \Lambda - \Lambda^{-1}{\rm diag}(a_{n})=\pi_{as}(L)$. 
Since all $a_{n}>0$ the diagonal matrix ${\rm diag}(a_{n})$ belongs to $\mathcal{D}_{1}(R)^{*}$ and it was shown in \cite{HPP1} that any matrix in $LT_{\mathbb{Z}}(R)$ of degree one with an invertible leading diagonal component can be obtained by dressing $\Lambda$ with an element of $\mathcal{P}_{\leqslant 0}$. This holds also for $F_{11}$. Hence $L$ can be obtained by dressing $F_{11}$ with an element of $\mathcal{P}_{\leqslant 0}$.
One can express the equations (\ref{1.2}) in terms of relations for the diagonal components before the different powers of $\Lambda$ in $L$ in (\ref{decoL1}). Thus we get:
\begin{align*}
\frac{d}{dt}({\rm diag}(a_{n}))&={\rm diag}(a_{n})({\rm diag}(p_{n})- \Lambda {\rm diag}(p_{n}) \Lambda^{-1})\\
&={\rm diag}(a_{n})({\rm diag}(p_{n})- {\rm diag}(p_{n+1}) )\\
\frac{d}{dt}({\rm diag}(p_{n}))&=2(\Lambda^{-1} {\rm diag}(a_{n}^{2}) \Lambda -{\rm diag}(a_{n}^{2}))\\
&=2({\rm diag}(a_{n-1}^{2})-{\rm diag}(a_{n}^{2}))
\end{align*}
and comparing these formulas with the expression of 
$$
[L,M]=[{\rm diag}(p_{n}) \Id +\Lambda^{-1} {\rm diag}(2a_{n}), {\rm diag}(a_{n}) \Lambda - \Lambda^{-1}{\rm diag}(a_{n})]$$ 
in diagonal matrices and powers of $\Lambda$ yields that 
the following Lax equation for $L$:
\begin{equation}
\label{inftoda}
\frac{dL}{dt}=[L,M],
\end{equation}
is equivalent with the equations (\ref{1.2}) of the infinite Toda chain.
\end{example}

The solutions to both hierarchies possess still another useful property.
There holds namely
\begin{proposition}
\label{P2.1}
Both the Lax equations (\ref{Laxsh}) of the $\mathcal{C}_{sh}(\mathbb{C})$-hierarchy and those (\ref{Laxas}) of the $\mathcal{C}_{as}(\mathbb{R})$-hierarchy are so-called compatible systems, i.e. the projections 
$\{ \mathcal{B}_{\sigma}:=\pi_{sh}(\mathcal{G}_{\sigma}) \mid \sigma \in \Sigma\}$ of a solution $\{\mathcal{G}_{\sigma} \}$  of the $\mathcal{C}_{sh}(\mathbb{C})$-hierarchy satisfy the zero curvature relations
\begin{equation}
\label{ZCSH}
\partial_{\sigma_{1}}(\mathcal{B}_{\sigma_{2}}
)-\partial_{\sigma_{2}}(\mathcal{B}_{\sigma_{1}})-[\mathcal{B}_{\sigma_{1}},\mathcal{B}_{\sigma_{2}}]=0
\end{equation}
and the projections $\{ \mathcal{C}_{j}:=\pi_{as}(\mathcal{F}_{j}) \mid j \in \mathcal{J} \}$ of a solution $\{\mathcal{F}_{j} \}$  of the $\mathcal{C}_{as}(\mathbb{R})$-hierarchy satisfy the zero curvature relations
\begin{equation}
\label{ZCAS}
\partial_{j_{1}}( \mathcal{C}_{j_{2}})-\partial_{j_{2}}( \mathcal{C}_{j_{1}})-[ \mathcal{C}_{j_{1}}
, \mathcal{C}_{j_{2}}]=0.
\end{equation}
\end{proposition}
\begin{proof}
The idea is to show that the left hand side of (\ref{ZCSH}) resp. (\ref{ZCAS}) belongs to 
$$
\pi_{sh}(LT_{\mathbb{Z}}(R)) \cap \pi_{-}(LT_{\mathbb{Z}}(R)) 
\text{ resp. }
\pi_{as}(LT_{\mathbb{Z}}(R)) \cap \pi_{\leqslant 0}(LT_{\mathbb{Z}}(R)) 
$$ 
and thus has to be zero. We give the proof for the $\{ \mathcal{B}_{\sigma} \}$, that for the $\{ \mathcal{C}_{j} \}$ is similar and is left to the reader. The inclusion in the first factor is clear as both $\mathcal{B}_{\sigma}$ and $\partial_{\sigma_{1}}(\mathcal{B}_{\sigma_{2}})$ belong to the Lie subalgebra $\pi_{sh}(LT_{\mathbb{Z}}(R)).$ 
To show the other one, we use the Lax equations (\ref{Laxsh}). 
By substituting $\mathcal{B}_{\sigma_{k}}=\mathcal{G}_{\sigma_{k}}-\pi_{-}(\mathcal{G}_{\sigma_{k}}), k=1,2,$
we get for
\begin{align*}
\partial_{\sigma_{1}}( \mathcal{B}_{\sigma_{2}})-\partial_{\sigma_{2}}( \mathcal{B}_{\sigma_{1}})=&\;\partial_{\sigma_{1}}(\mathcal{G}_{\sigma_{2}})-\partial_{\sigma_{1}}(\pi_{-}(\mathcal{G}_{\sigma_{2}}))\\
 &\; -\partial_{\sigma_{2}}(\mathcal{G}_{\sigma_{1}})+\partial_{\sigma_{2}}(\pi_{-}(\mathcal{G}_{\sigma_{1}}))\\
=&\; [\mathcal{B}_{\sigma_{1}},\mathcal{G}_{\sigma_{2}}]-[\mathcal{B}_{\sigma_{2}},\mathcal{G}_{\sigma_{1}}]\\
& \;-\partial_{\sigma_{1}}(\pi_{-}(\mathcal{G}_{\sigma_{2}}))+\partial_{\sigma_{2}}(\pi_{-}(\mathcal{G}_{\sigma_{1}}\mathcal{G}_{\sigma_{k}}))
\end{align*} 
and for 
\begin{align*}
[ \mathcal{B}_{\sigma_{1}}, \mathcal{B}_{\sigma_{2}}]=&\;[\mathcal{G}_{\sigma_{1}}-\pi_{-}(\mathcal{G}_{\sigma_{1}}), \mathcal{G}_{\sigma_{2}}-\pi_{-}(\mathcal{G}_{\sigma_{2}})]\\
=&\;-[\pi_{-}(\mathcal{G}_{\sigma_{1}}), \mathcal{G}_{\sigma_{2}}]+[\pi_{-}(\mathcal{G}_{\sigma_{2}}), \mathcal{G}_{\sigma_{1}}]\\
&\;+[\pi_{-}(\mathcal{G}_{\sigma_{1}}),\pi_{-}(\mathcal{G}_{\sigma_{2}})].
\end{align*}
Taking into account the second identity in (\ref{Laxsh}), we see that the left hand side of (\ref{ZCSH}) is equal to 
$$
-\partial_{\sigma_{1}}(\pi_{-}(\mathcal{G}_{\sigma_{2}}))+\partial_{\sigma_{2}}(\pi_{-}(\mathcal{G}_{\sigma_{1}}))-[\pi_{-}(\mathcal{G}_{\sigma_{1}}),\pi_{-}(\mathcal{G}_{\sigma_{2}})].
$$
This element belongs to the Lie subalgebra $\pi_{-}(LT_{\mathbb{Z}}(R)) $
and that proves the claim.
\end{proof}

Besides the zero curvature relations for the projections $\{ \mathcal{B}_{\sigma} \}$ resp. $\{ \mathcal{C}_{j} \}$ corresponding to respectively a solution $\{\mathcal{G}_{\sigma} \}$  of the $\mathcal{C}_{sh}(\mathbb{C})$-hierarchy and a solution $\{\mathcal{F}_{j} \}$  of the $\mathcal{C}_{as}(\mathbb{R})$-hierarchy, also other parts satisfy such relations. Introduce for any $\mathcal{P}_{-}(\mathbb{R})$-deformation $\{\mathcal{G}_{\sigma} \}$ of the $\{G_{\sigma} \}$
and any  $\mathcal{P}_{\leqslant 0}$-deformation $\{\mathcal{F}_{j} \}$ of the $\{F_{j} \}$
in $LT_{\mathbb{Z}}(R)$ the notations
\begin{equation}
\label{co's}
\mathcal{A}_{\sigma}:=\mathcal{B}_{\sigma}-\mathcal{G}_{\sigma}=-\pi_{-}(\mathcal{G}_{\sigma} ) , \sigma \in \Sigma, \text{ and }\mathcal{D}_{j}:=\mathcal{C}_{j}-\mathcal{F}_{j}=-\pi_{\leqslant 0}(\mathcal{F}_{j}), j \in \mathcal{J}.
\end{equation}
Then we can say

\begin{corollary}The following relations hold:
\label{C1.1}
\begin{itemize}
\item The parts $\{ \mathcal{A}_{\sigma} \mid \sigma \in \Sigma\}$ of a solution $\{\mathcal{G}_{\sigma} \}$ of the $\mathcal{C}_{sh}(\mathbb{C})$-hierarchy satisfy 
$$
\partial_{\sigma_{1}}(\mathcal{A}_{\sigma_{2}})-\partial_{\sigma_{2}}(\mathcal{A}_{\sigma_{1}}
)-[\mathcal{A}_{\sigma_{1}},\mathcal{A}_{\sigma_{2}}]=0.
$$
\item
The parts $\{ \mathcal{D}_{j} \mid j \in \mathcal{J}\}$ of a solution $\{\mathcal{F}_{j} \}$ of the $\mathcal{C}_{as}(\mathbb{R})$-hierarchy satisfy 
$$
\partial_{j_{1}}(\mathcal{D}_{j_{2}})-\partial_{j_{2}}(\mathcal{D}_{j_{1}})-[\mathcal{D}_{j_{1}},\mathcal{D}_{j_{2}}]=0
$$
\end{itemize}
\end{corollary}
\begin{proof}
Again we show the result only for the $\mathcal{C}_{sh}(\mathbb{C})$-hierarchy. 
Now we substitute in the zero curvature relations for the $\{ \mathcal{B}_{\sigma}\}$ everywhere the relation $\mathcal{B}_{\sigma}=\mathcal{A}_{\sigma}+\mathcal{G}_{\sigma}$, use the second equality in the Lax equations (\ref{Laxsh}) and the fact that all the $\{\mathcal{G}_{\sigma}\}$ commute. This gives the desired result.
\end{proof}

Let $(R,\{\partial_{\sigma} \})$ 
denote a setting for the $\mathcal{C}_{sh}(\mathbb{C})$-hierarchy and we write $(R_{as},\{\partial_{j} \})$ for a setting of the $\mathcal{C}_{as}(\mathbb{R})$-hierarchy. 
For both hierarchies 
there also holds an analogue of the Sato-Wilson equations for the KP hierarchy 
\begin{proposition}
\label{P5.1} Let $\{\mathcal{G}_{\sigma} =KG_{\sigma} K^{-1}\}$ be a $\mathcal{P}_{-}(\mathbb{R})$-deformation of the $\{ G_{\sigma} \}$ in $LT_{\mathbb{Z}}(R)$ and likewise let $\{\mathcal{F}_{j} =PF_{j} P^{-1}\}$ be a  $\mathcal{P}_{\leqslant 0}$-deformation of the $\{ F_{j}\}$ in $LT_{\mathbb{Z}}(R_{as})$. Define the $\{\mathcal{A}_{\sigma}\}$ and the $\{ \mathcal{D}_{j}\}$ as in (\ref{co's}). Then there holds:
\begin{enumerate}
\item[(a)] If the dressing operator $K$ of the $\{\mathcal{G}_{\sigma}\}$ satisfies the equations
\begin{equation}
\label{SW1}
\partial_{\sigma}(K)=\mathcal{A}_{\sigma}K, \text { for all } \sigma \in \Sigma,
\end{equation}
then $\{\mathcal{G}_{\sigma} \}$ is a solution of the $\mathcal{C}_{sh}(\mathbb{C})$-hierarchy. Note that it makes sense to consider equations (\ref{SW1}), because both sides have degree zero in $\Lambda$ or lower.
\item[(b)] Similarly, if the dressing operator $P$ of the $\{\mathcal{F}_{j}\}$ satisfies the equations
\begin{equation}
\label{SW2}
\partial_{j}(P)=\mathcal{D}_{j}P, \text { for all } j \in \mathcal{J},
\end{equation}
then $\{\mathcal{F}_{j} \}$ is a solution of the $\mathcal{C}_{as}(\mathbb{R})$-hierarchy. Note that both sides of the equations (\ref{SW2}) are of degree zero in $\Lambda$ or lower. 
\end{enumerate}
The equations (\ref{SW1}) resp. (\ref{SW2}) are called the Sato-Wilson equations of the $\mathcal{C}_{sh}(\mathbb{C})$-hierarchy resp. the $\mathcal{C}_{as}(\mathbb{R})$-hierarchy.
\end{proposition}
\begin{proof}
We present again the proof for the case (\ref{SW1}) that for (\ref{SW2}) is similar. Since  $G_{\sigma_{1}}$ is constant w.r.t. the $\{ \partial_{\sigma} \}$, we have 
in general
$$
\partial_{\sigma_{1}}(KG_{\sigma_{2}}K^{-1})=\partial_{\sigma_{1}}(K)K^{-1}KG_{\sigma_{2}}K^{-1}-KG_{\sigma_{2}}K^{-1}\partial_{\sigma_{1}}(K)K^{-1}=[\partial_{\sigma_{1}}(K)K^{-1},\mathcal{G}_{\sigma_{2}}].
$$
According to (\ref{SW1}) each $\partial_{\sigma_{1}}(K)K^{-1}$ equals $\mathcal{A}_{\sigma_{1}}$ and substitution in the last equation results in
$$
\partial_{\sigma_{1}}(\mathcal{G}_{\sigma_{2}})=[\mathcal{A}_{\sigma_{1}},\mathcal{G}_{\sigma_{2}}]=[\pi_{sh}(\mathcal{G}_{\sigma_{1}})-\mathcal{G}_{\sigma_{1}},\mathcal{G}_{\sigma_{2}}]=[\pi_{sh}(\mathcal{G}_{\sigma_{1}}),\mathcal{G}_{\sigma_{2}}],
$$
which proves the claim.
\end{proof}
In both cases solutions of the same system differ by an element of the dressing group that is constant. We illustrate that for system (\ref{SW1}).
If $K_{1}$ is another solution of (\ref{SW1}) for the same set of $\{ \mathcal{A}_{\sigma} \}$, then $K_{1}=KK_{0}$, where $K_{0}$ is a matrix in $LT_{\mathbb{Z}}(R)$
that is constant for all the $\{ \partial_{\sigma}\}$, i.e. $\partial_{\sigma}(K_{0})=0$, for all $\sigma \in \Sigma$. There holds namely,
$$
\partial_{\sigma}(K_{1})=\mathcal{A}_{\sigma}K_{1}=\partial_{\sigma}(K)K^{-1}K_{1}+K \partial_{\sigma}(K_{0})=\mathcal{A}_{\sigma}K_{1}+K \partial_{\sigma}(K_{0}).
$$
This implies $K \partial_{\sigma}(K_{0})=0$ and, as $K$ is invertible, the desired identity holds. Reversely, for any  
$K_{0} \in P_{-}(\mathbb{R})$ satisfying $\partial_{\sigma}(K_{0})=0$ for all $\sigma \in \Sigma$
and any solution $K$ of system (\ref{SW1}), the operator $KK_{0}$ is another  solution of (\ref{SW1}). Proposition \ref{P5.1} offers the possibility to construct solutions of respectively  the $\mathcal{C}_{sh}(\mathbb{C})$-hierarchy
and the $\mathcal{C}_{as}(\mathbb{R})$-hierarchy by finding dressing matrices that satisfy the equations (\ref{SW1}) resp. (\ref{SW2}). This can be achieved  for each of the two hierarchies by constructing special vectors, called wave matrices, in appropriate left $LT_{\mathbb{Z}}(R)$- and $LT_{\mathbb{Z}}(R_{as})$-modules that satisfy a set of relations, the linearization of the hierarchy in question, in this module. This will be the topic of the next section.

\section{Linearizations and wave matrices } 
\label{sec4}

We start out in this section with a $\mathcal{P}_{-}(\mathbb{R})$-deformation $\{\mathcal{G}_{\sigma}\}$ of the $\{ G_{\sigma} \}$ together with the projections $\{ \mathcal{B}_{\sigma}:=\pi_{sh}(\mathcal{G}_{\sigma}) \mid \sigma \in \Sigma\}$ in the setting $(R,\{\partial_{\sigma} \})$ and a  $\mathcal{P}_{\leqslant 0}$-deformation $\{\mathcal{F}_{j} \}$ of the $\{ F_{j}\}$ together with the projections $\{ \mathcal{C}_{j}:=\pi_{as}(\mathcal{F}_{j}) \mid j \in \mathcal{J} \}$ in the setting $(R_{as},\{\partial_{j} \})$.   Before giving the formal description of each of the two modules, we first present two sets of equations for respectively the 
$\{ \mathcal{G}_{\sigma}  \}$ and the 
$\{ \mathcal{F}_{j}  \}$ and a number of manipulations with these equations that lead to the Lax equations (\ref{Laxsh}) respectively (\ref{Laxas}) for them. Like this it will be clear that these sets of equations are connected to the hierarchy for which they deliver the Lax equations. Later on, we present the formal framework in which these manipulations are justified.
For the $\{ \mathcal{G}_{\sigma}  \}$ one searches an appropriate $\varphi$ in the module such that the following  set of equations holds in the module 
\begin{align}
\label{LinG1}
&\mathcal{G}_{\sigma} \varphi=\varphi 
G_{\sigma} \text{ for all } \sigma \in \Sigma ,\\
\label{LinG2}
&\partial_{\sigma}(\varphi)=\pi_{sh}(\mathcal{G}_{\sigma} )\varphi  \text{ for all }\sigma \in \Sigma  .
\end{align}
This set is called the {\it linearization of the $\mathcal{C}_{sh}(\mathbb{C})$-hierarchy}. 
 To get the Lax equations of the $\mathcal{C}_{sh}(\mathbb{C})$-hierarchy we apply $\partial_{\sigma_{1}}$ to the equation (\ref{LinG1}) for $\sigma_{2}$, use a Leibnitz rule for products and substitute twice equation (\ref{LinG2}). 
This yields
\begin{align}
&\partial_{\sigma_{1}}(\mathcal{G}_{\sigma_{2}} \varphi - \varphi G_{\sigma_{2}}) =\partial_{\sigma_{1}}(\mathcal{G}_{\sigma_{2}} )\varphi +
\mathcal{G}_{\sigma_{2}} (\partial_{\sigma_{1}}(\varphi)) - (\partial_{\sigma_{1}}(\varphi)) G_{\sigma_{2}}  & \label{linGLax}\\ \notag
&=\partial_{\sigma_{1}}(\mathcal{G}_{\sigma_{2}} )\varphi + \mathcal{G}_{\sigma_{2}} \pi_{as}(\mathcal{G}_{\sigma_{1}} ) \varphi -\pi_{as}(\mathcal{G}_{\sigma_{1}} ) \varphi  G_{\sigma_{2}}\\ \notag
&= \{\partial_{\sigma_{1}}(\mathcal{G}_{\sigma_{2}} )-[ \pi_{as}(\mathcal{G}_{\sigma_{1}} ),\mathcal{G}_{\sigma_{2}} \;] \}
\varphi =0. &
\label{}
\end{align}
Hence, if the annihilator of $\varphi$ in the $LT_{\mathbb{Z}}(R)$-module is equal to zero, then the $\{ \mathcal{G}_{\sigma}  \}$ satisfy the Lax equations of the $\mathcal{C}_{sh}(\mathbb{C})$-hierarchy.

\begin{remark}
\label{R3.1}
The form of the {\it linearization of the $\mathcal{C}_{as}(\mathbb{R})$-hierarchy} is as follows:
\begin{align}
\label{LinF1}
&\mathcal{F}_{j} \psi=\psi F_{j} \text{ for all } j  \in \mathcal{J} ,\\
\label{LinF2}
&\partial_{j}(\psi)=\pi_{as}(\mathcal{F}_{j} )\psi  \text{ for all } j \in \mathcal{J}.
\end{align}
Here $\psi$ is a vector in the left $LT_{\mathbb{Z}}(R_{as})$-module corresponding to the $\mathcal{C}_{as}(\mathbb{R})$-hierarchy. A similar set of manipulations as above yields then the Lax equations for the $\{ \mathcal{F}_{j}  \}$, if the annihilator of $\psi$ is zero.
\end{remark}

Now we discuss the $LT_{\mathbb{Z}}(R)$-module for the $\mathcal{C}_{sh}(\mathbb{C})$-hierarchy.
Recall from the manipulations that we need in (\ref{LinG1}) and (\ref{LinG2}) that one can multiply $\varphi $ from the left with elements from $LT_{\mathbb{Z}}(R)$ like $\mathcal{G}_{\sigma}$ and 
 $\pi_{sh}(\mathcal{G}_{\sigma} )$ and from the right with all the basic matrices $\{ G_{\sigma}\}$. Further, there should be a left action of each
 $\partial_{\sigma}$ on $\varphi$ that satisfies Leibnitz with respect to the left $LT_{\mathbb{Z}}(R)$-action and finally the annihilator of $\varphi$ should be zero. To realize the first, one builds a suitable left $LT_{\mathbb{Z}}(R)$--module, where also the other actions can be given sense. The actual form of the elements
in the module is guided by the $\varphi_{0}$ corresponding to the trivial solution $\{ \mathcal{G}_{\sigma}=G_{\sigma} \}$
of the
hierarchy. In that case the equations (\ref{LinG1}) and (\ref{LinG2}) reduce to
\begin{align}\label{trivlin}
G_{\sigma}\varphi_{0}=\varphi_{0} G_{\sigma}, 
  \text{ and }
\partial_{\sigma }(\varphi_{0})=G_{\sigma} \varphi_{0}.
\end{align}
If one thinks of $\varphi_{0}$ as a $\mathbb{Z} \times \mathbb{Z}$-matrix then the first equation in (\ref{trivlin}) tells you that $\varphi_{0}$ commutes with all $\{G_{\sigma}\}$ and, since $\partial_{\sigma}$ is the algebraic substitute for differentiation w.r.t. the flow parameter corresponding to $G_{\sigma}$, the second equation of (\ref{trivlin}) yields for all $\sigma \in \Sigma$ that $\tilde{\varphi}_{0}:=\exp(-t_{\sigma}G_{\sigma})\varphi_{0}$ is constant for $\partial_{\sigma}$, i.e. 
$\partial_{\sigma}(\tilde{\varphi}_{0})=0$.  
This leads one to consider the formal series
\begin{equation}\label{}
\varphi_0:=\exp( \sum_{\sigma \in \Sigma}  t_{\sigma} G_{\sigma}).
\end{equation}
Under suitable convergence conditions, see Section \ref{sec5}, this series corresponds to a well-defined $\mathbb{Z} \times \mathbb{Z}$-matrix, it commutes with all the $\{ G_{\sigma}\}$ and, if one lets $\partial_{\sigma}$ act on $\varphi_{0}$ as $\frac{\partial}{\partial t_{\sigma }}$, then it satisfies the second equation in (\ref{trivlin}).
The module for the linearization will consist of formal perturbations of this trivial solution $\varphi_{0}$ by formal multiplication with elements from $LT_{\mathbb{Z}}(R)$ from the left.  Consider namely the collection $\mathcal{O}_{sh}$ 
of formal products
\begin{equation}\label{Msh}
\{ \sum_{r=-\infty}^{N } d_r\Lambda^{r} \} \exp( \sum_{\sigma \in \Sigma} t_{\sigma}  G_{\sigma} ) ,\text{ where all }d_{r}\in \mathcal{D}_1(R). 
\end{equation}
Notice that, even if $\varphi_0$ is a well-defined $\mathbb{Z} \times \mathbb{Z}$-matrix, then the product in (\ref{Msh}) of the perturbation factor from $LT_{\mathbb{Z}}(R)$ and $\varphi_0$ is in general not a well-defined $\mathbb{Z} \times \mathbb{Z}$-matrix. So, it is necessary to keep the two factors separate.
Analogously to the terminology used at the KP hierarchy, we define
\begin{definition}
\label{defomi}
The elements of $\mathcal{O}_{sh}$ are called {\it oscillating matrices} for the $\mathcal{C}_{sh}(\mathbb{C})$-hierarchy.
\end{definition}
Despite the fact that the product in (\ref{Msh}) is formal, there is a
well-defined left action of
$LT_{\mathbb{Z}}(R)$ on it. For all $\ell_1$ and $\ell_2 \in LT_{\mathbb{Z}}(R)$ one puts namely
\begin{equation}\label{5.14}
 \ell_1 \{ \ell_2 \}\varphi_0=\{\ell_1 \ell_2
\}\varphi_0.
\end{equation}
Also the right multiplication with $\{ G_{\sigma} \}$ 
is well-defined on elements of $\mathcal{O}_{sh}$
\begin{equation}
\label{ractie}
\{ \ell \}\varphi_0G_{\sigma}:=\{ \ell G_{\sigma} \}\varphi_0.
\end{equation}
An action of the derivations $\partial_{\sigma}$ on $\mathcal{O}_{sh}$ is defined as if the product in the module $\mathcal{O}_{sh}$ is a real one
\begin{equation}\label{derMsh}
\partial_{\sigma}( \{ \sum_{j=-\infty}^{N } d_j \Lambda^{j} \} \varphi_0 )=  \{ \sum_{j=-\infty}^{N } \partial_{\sigma}(d_j) \Lambda^{j} +\sum_{j=-\infty}^{N } d_j \Lambda^{j} G_{\sigma}
 \} \varphi_0.
\end{equation}
It is a direct verification that this action of $\partial_{\sigma}$ satisfies the Leibnitz rule used in the manipulations in (\ref{linGLax}).
Note
that $\mathcal{O}_{sh}$ is a free $LT_{\mathbb{Z}}(R)$-module with generator $\varphi_0$. Hence scratching $\varphi$ from the equations (\ref{LinG1}) and (\ref{LinG2})
is
permitted as soon as one knows that $\varphi=\hat{\varphi} \varphi_0$ with $\hat{\varphi} \in I(LT_{\mathbb{Z}}(R))$. 
Moreover, the equations
$\mathcal{G}_{\sigma} \varphi = \varphi  G_{\sigma}$ imply then that 
$\mathcal{G}_{\sigma}=\hat{\varphi}G_{\sigma} \hat{\varphi}^{-1}
$. Since we are interested in $\mathcal{P}_{-}(\mathbb{R})$-deformations of the $\{G_{\sigma}\}$, this brings us to the following 
\begin{definition}
\label{defwmi}
An oscillating 
matrix 
$\varphi=\hat{\varphi} \varphi_0$, with $\hat{\varphi} \in \mathcal{P}_{-}(R)$,
is
called {\it a wave matrix } 
for the matrices $\{ \mathcal{G}_{\sigma}  \}$, if $\varphi$ and the $\{ \mathcal{G}_{\sigma}  \}$ satisfy the equations (\ref{LinG1}) and (\ref{LinG2}). 
\end{definition}
Since the manipulations to get the Lax equations are well-defined on
such a $\varphi$, the $\{ \mathcal{G}_{\sigma}  \}$
form a solution of the $\mathcal{C}_{sh}(\mathbb{C})$-hierarchy.  This follows also from Proposition \ref{P5.1}, as $\hat{\varphi}$ satisfies the Sato-Wilson equations (\ref{SW1}) of this hierarchy. Substitute namely $\varphi=\hat{\varphi} \varphi_0$ in equation (\ref{LinG2}) and one gets 
$$
\partial_{\sigma}(\varphi)=(\partial_{\sigma}(\hat{\varphi})\hat{\varphi}^{-1}+\hat{\varphi}G_{\sigma}\hat{\varphi}^{-1})\varphi=\mathcal{B}_{\sigma}\varphi
$$
and, since  $\mathcal{O}_{sh}$ is a free $LT_{\mathbb{Z}}(R)$-module with generator $\varphi$, one obtains 
$\partial_{\sigma}(\hat{\varphi})\hat{\varphi}^{-1}
=\mathcal{A}_{\sigma}$.

\begin{remark}
\label{R3.2} For the form of the $LT_{\mathbb{Z}}(R_{as})$-module one looks also first at the solution $\psi_{0}$ of the linearization of the $\mathcal{C}_{as}(\mathbb{R})$-hierarchy corresponding to the trivial solution $\{\mathcal{F}_{j}=F_{j}\}$. Thus one obtains the series 
$$
\psi_{0}=\exp(\sum_{j \in \mathcal{J}}t_{j}F_{j}),
$$ 
which determines under suitable convergence conditions a well-defined $\mathbb{Z} \times \mathbb{Z}$-matrix. The module $\mathcal{O}_{as}$ of {\it oscillating matrices} for the $\mathcal{C}_{as}(\mathbb{R})$-hierarchy consists of all formal products of a perturbation factor from $LT_{\mathbb{Z}}(R_{as})$ and $\psi_{0}$. If you replace in (\ref{5.14}) everywhere $\varphi_{0}$ by $\psi_{0}$ and $LT_{\mathbb{Z}}(R)$ by $LT_{\mathbb{Z}}(R_{as})$, then you get the $LT_{\mathbb{Z}}(R_{as})$-module structure on $\mathcal{O}_{as}$. The right action of the basis of $\mathcal{C}_{as}(\mathbb{R})$ you obtain by altering in (\ref{ractie}) each $\varphi_{0}$ into a $\psi_{0}$ and each $G_{\sigma}$ into an $F_{j}$. Finally, the action of the derivations $\{ \partial_{j}\}$ on $\mathcal{O}_{as}$ you obtain if you replace in (\ref{derMsh}) each $\partial_{\sigma}$ by a $\partial_{j}$, each $\varphi_{0}$ by a $\psi_{0}$ and 
each $G_{\sigma}$ by an $F_{j}$. $\mathcal{O}_{as}$ is a free $LT_{\mathbb{Z}}(R_{as})$-module with generator $\psi_0$. The elements $\psi=\hat{\psi}\psi_{0}$ in $\mathcal{O}_{as}$, where $\hat{\psi}$ is invertible in $LT_{\mathbb{Z}}(R_{as})$, have a zero annihilator. Let $\psi=\hat{\psi}\psi_{0}$ be an element of $\mathcal{O}_{as}$ such that $\hat{\psi} \in \mathcal{P}_{\leqslant 0}$. Then $\psi$ is called a {\it wave matrix} for the 
$\{\mathcal{F}_{j}=\hat{\psi}F_{j}\hat{\psi}^{-1} \}$, if $\psi$ and the $\{\mathcal{F}_{j} \}$ satisfy the equations (\ref{LinF1}) and (\ref{LinF2}). In particular, the $\{\mathcal{F}_{j} \}$ are then a solution of the $\mathcal{C}_{as}(\mathbb{R})$-hierarchy. This is once more confirmed by the fact that the perturbation factor $\hat{\psi}$ of the wave matrix of the $\{\mathcal{F}_{j} \}$ satisfies the Sato-Wilson equations (\ref{SW2}).
\end{remark}

If one wants to prove
the equations (\ref{LinG1}) and (\ref{LinG2}) for an oscillating matrix $\varphi$ for the $\mathcal{C}_{sh}(\mathbb{C})$-hierarchy
 of the right form or the equations (\ref{LinF1}) and (\ref{LinF2}) for an oscillating matrix $\psi$ for the $\mathcal{C}_{as}(\mathbb{R})$-hierarchy
 of the appropriate shape, it suffices to prove weaker results, as the next Proposition demonstrates
\begin{proposition}\label{Pweak} 
For each hierarchy we have
\begin{enumerate}
\item[(a)] 
Let $\varphi=\hat{\varphi} \varphi_0$, with $\hat{\varphi} \in \mathcal{P}(\mathbb{R})$,
be an oscillating matrix for the $\mathcal{C}_{sh}(\mathbb{C})$-hierarchy. If it satisfies for all $\sigma \in \Sigma$ 
\begin{equation*}\notag
\partial_{\sigma}(\varphi)=S_{\sigma} \phi, \text{ with } S_{\sigma } \in \mathcal{SH}(R),
\end{equation*}
then $S_{\sigma}=\pi_{sh}(\mathcal{G}_{\sigma})$, where $\mathcal{G}_{\sigma}:=\hat{\varphi} ( G_{\sigma} ) \hat{\varphi}^{-1}$.
In particular the $\{ \mathcal{G}_{\sigma}  \}$ form a
solution to the
$\mathcal{C}_{sh}(\mathbb{C})$-hierarchy  and $\varphi$ is a wave matrix 
for this solution.
\item[(b)]Let $\psi=\hat{\psi} \psi_0$, with $\hat{\varphi} \in \mathcal{P}_{\leqslant 0}$,
be an oscillating matrix for the $\mathcal{C}_{as}(\mathbb{R})$-hierarchy. If it satisfies for all $j \in \mathcal{J}$ 
\begin{equation*}\notag
\partial_{j}(\psi)=A_{j} \psi, \text{ with } A_{j } \in \mathcal{AS}(R_{as}),
\end{equation*}
then $A_{j}=\pi_{as}(\mathcal{F}_{j})$, where $\mathcal{F}_{j}:=\hat{\psi} ( F_{j} ) \hat{\psi}^{-1}$
In particular the $\{ \mathcal{F}_{j}  \}$ form a
solution to the
$\mathcal{C}_{as}(\mathbb{R})$-hierarchy  and $\psi$ is a wave matrix 
for this solution.
\end{enumerate}
\end{proposition}
\begin{proof}
We just give the proof for part (a), that of (b) is similar and left to the reader.
From the definition of the action of $\partial_{\sigma}$ on $\mathcal{O}_{sh}$ and the fact that $\mathcal{O}_{sh}$ is a free $LT_{\mathbb{Z}}(R)$-module with generator $\varphi_0$, one gets the operator equation
\begin{equation}\label{5.16}
\partial_{\sigma}(\hat{\varphi} ) - \hat{\varphi}  G_{\sigma} =S_{\sigma} \hat{\varphi}.
\end{equation}
Multiplying this equation from the
right with $ \hat{\varphi}^{-1}$ and applying the projection $\pi_{sh}$
gives the desired result. 
\end{proof} 
\begin{remark}
The wave matrices for the $\mathcal{C}_{sh}(\mathbb{C})$-hierarchy and those for the $\mathcal{C}_{as}(\mathbb{R})$-hierarchy are in their context the analogues of the Baker-Akhiezer functions for the $KP$-hierarchy.
\end{remark}

It might happen for both the wave matrices for the $\mathcal{C}_{sh}(\mathbb{C})$-hierarchy as for those of the $\mathcal{C}_{as}(\mathbb{R})$-hierarchy that different wave matrices give the same solution of the hierarchy. We discuss here the freedom one has and we start with two wave matrices $\varphi_1=\hat{\varphi}_{1}\varphi_{0}$ and
$\varphi_2=\hat{\varphi}_{2}\varphi_{0}$ that give the same solution of the $\mathcal{C}_{sh}(\mathbb{C})$-hierarchy.
So both $\hat{\phi}_{1}$ and $\hat{\phi}_{2}$ belong to $\mathcal{P}_{-}(\mathbb{R})$ and there holds for all $\sigma \in \Sigma$ that 
$$
\mathcal{G}_{\sigma}=\hat{\varphi}_1  
G_{\sigma} \hat{\varphi}_{1}^{-1}=\hat{\varphi}_2 G_{\sigma} \hat{\phi_2}^{-1}.
$$ 
Then one has first of all that $\hat{\varphi}_1^{-1}\hat{\varphi}_2$ commutes with all the $\{ G_{\sigma}), \sigma \in \Sigma.$ 
As we have seen in the proof of Lemma \ref{L2.2}, commuting in $LT_{\mathbb{Z}}(R)$ with the basis of $\mathcal{C}_{sh}(\mathbb{C})$ is equivalent with commuting with $\Lambda$. 
Therefore we get
$$
\hat{\varphi}_1^{-1}\hat{\varphi}_2=\sum_{i \leqslant 0} j_{1}(a_{i}) \Lambda^{i} \in \mathcal{P}_{-}(\mathbb{R}).
$$
Hence $a_{0} \in R_{as}^{*}, a_{0} >0,$ and all other $a_i \in R$. 
One has seen in the proof of
Proposition \ref{Pweak} that for all
$\sigma \in \Sigma$
and $i=1,2,$ there holds
$$
\partial_{\sigma }(\hat{\varphi}_i )=\pi_{sh}(\mathcal{G}_{\sigma}) \hat{\varphi}_i -\hat{\varphi}_i
G_{\sigma}.
$$
Hence, if one applies the operator $\partial_{\sigma}$ to the equality $\hat{\varphi}_2=\hat{\varphi}_1 \sum_{i \leqslant 0} j_{1}(a_{i}) \Lambda^{i} $, then one obtains
\begin{align}
\partial_{\sigma}(\hat{\varphi}_2)=\partial_{\sigma}(\hat{\varphi}_1)\sum_{i \leqslant 0} j_{1}(a_{i}) \Lambda^{i} +
\hat{\phi}_1\sum_{i \leqslant 0} i_{1}(\partial_{j}(a_{i})) \Lambda^{i}=
\label{5.16.1}\\
\pi_{sh}(\mathcal{G}_{\sigma}) \hat{\varphi}_1-\hat{\varphi}_1G_{\sigma} \sum_{i \leqslant 0} i_{1}(a_{i}) \Lambda^{i} +
\hat{\varphi}_1 \sum_{i \leqslant 0} j_{1}(\partial_{\sigma}(a_{i})) \Lambda^{i}= \label{5.16.2}\\
\pi_{sh}(\mathcal{G}_{\sigma}) \hat{\varphi}_2-\hat{\varphi}_2 G_{\sigma} + \hat{\varphi}_1\sum_{i \leqslant 0} j_{1}(\partial_{\sigma}(a_{i})) \Lambda^{i}=\\
\partial_{\sigma}(\hat{\varphi}_2)+\hat{\varphi}_1\sum_{i \leqslant 0} j_{1}(\partial_{\sigma}(a_{i})) \Lambda^{i}
\end{align}
Hence, the expression $\hat{\varphi}_1\sum_{i \leqslant 0} j_{1}(\partial_{\sigma}(a_{i})) \Lambda^{i}$ has to be zero.
Since  $\hat{\phi_1}$ is invertible in $LT_{\mathbb{Z}}(R)$, one must have for all $i \leqslant 0$ and  all $\sigma \in \Sigma$ 
that
$\partial_{j }( a_i)=0.$  One proceeds in the same way with two wave matrices $\psi_1=\hat{\psi}_{1}\psi_{0}$ and
$\psi_2=\hat{\psi}_{2}\psi_{0}$ that give the same solution of the $\mathcal{C}_{as}(\mathbb{R})$-hierarchy. We state the results in the following Corollary: 
\begin{corollary}\label{c5.2}
The freedom for each hierarchy in the wave matrices that yield the same solution is given by
\begin{enumerate}
\item[(a)] 
Assume $\varphi_1$ and $\varphi_2$ are wave matrices corresponding to the same solution $\{ \mathcal{G}_{\sigma}  \}$ of the $\mathcal{C}_{sh}(\mathbb{C})$-hierarchy.
Then there is an element $\sum_{i \leqslant 0} j_{1}(a_{i}) \Lambda^{i}$ in $\mathcal{P}_{-}(\mathbb{R})$
 such that
$$
\hat{\varphi}_2=\hat{\varphi}_1 \sum_{i \leqslant 0} j_{1}(a_{i}) \Lambda^{i}, \text{ with }a_{0} \in R(\mathbb{R})^{*}, a_{0} >0, \text{ and all other }a_{i} \in R.
$$
Moreover, all $a_{i}$ are constant for the derivations $\partial_{\sigma}$, i.e. $\partial_{\sigma}(a_{i})=0.$
\item[(b)] Assume that $\psi_1$ and $\psi_2$ are wave matrices corresponding to the same solution $\{ \mathcal{F}_{j}  \}$ of the $\mathcal{C}_{as}(\mathbb{R})$-hierarchy.
Then there is an element $\sum_{i \leqslant 0} j_{1}(b_{i}) \Lambda^{i}$ in $\mathcal{P}_{\leqslant 0}$ such that
$$
\hat{\psi}_2=\hat{\psi}_1 \sum_{i \leqslant 0} j_{1}(b_{i}) \Lambda^{i}, \text{ with }b_{0} \in R_{as}^{*}, b_{0} >0, \text{ and all other }b_{i} \in R_{as}.
$$
Moreover, all the $b_i$ are constant for the derivations $\partial_{j}$, i.e. $\partial_{j}(b_i)=0.$
\end{enumerate}
\end{corollary}
If one wants to use the set-up of oscillating matrices, linearizations and wave matrices, to construct concrete solutions of both hierarchies, then one has to take care that $\varphi_{0}$ and $\psi_{0}$ are 
well-defined $\mathbb{Z} \times \mathbb{Z}$-matrices and that the product between the perturbation factor and $\varphi_{0}$ or $\psi_{0}$ is also well-defined and yields a $\mathbb{Z} \times \mathbb{Z}$-matrix. In the next section we present for both hierarchies such a convergent framework in the style of \cite{Segal-Wilson}.

\section{The construction of solutions}
\label{sec5}

A natural way to obtain real or complex $\mathbb{Z} \times \mathbb{Z}$-matrices is the following: take a real or complex Hilbert space $H$ with a Hilbert basis $\{ e_{i} \mid i \in \mathbb{Z}\}$.
Then there corresponds to each bounded operator $b:H \to H$, a real or complex $\mathbb{Z} \times \mathbb{Z}$-matrix $[b]=(b_{ij})$ by the formula
$$
b(e_{j})=\sum_{i \in \mathbb{Z}} b_{ij} e_{i}.
$$
The wave matrices for both hierarchies that we will produce in this section, are constructed from a geometric context by using this principle. The relevant Hilbert spaces for this paper consist of the spaces $\mathcal{H}(k), k=\mathbb{R} \text{ or }\mathbb{C}$, of $\mathbb{Z} \times 1$-matrices with coefficients from $k$ given by
$$
\mathcal{H}(k)=\{ \vec{x} =\sum_{n \in \mathbb{Z}}x_{n} \vec{e}\,(n) \mid x_{n} \in k, \,  \sum_{n \in \mathbb{Z}} |x_{n}|^{2} < \infty \}.
$$
We put on $\mathcal{H}(k)$ the standard inner product 
$$
< \vec{x} \mid  \vec{y}>=\sum_{n \in \mathbb{Z}} x_{n}\overline{y}_{n}
$$
so that the $\{ \vec{e}\,(n) \mid n \in \mathbb{Z} \}$ form an orthonormal basis of $\mathcal{H}(k)$ and for each $b \in B(\mathcal{H}(k))$, the bounded $k$-linear operators from $\mathcal{H}(k)$ to itself, the matrix $[b]$ is taken w.r.t. this Hilbert basis. In particular, the action of $b$ on $\vec{x}$ is multiplying from the left, $M_{[b]}$, with 
$[b]$.
The commuting flows that play a role in the $\mathcal{C}_{sh}(\mathbb{C})$-hierarchy resp. $\mathcal{C}_{as}(\mathbb{R})$-hierarchy come from the commuting directions  $\{ G_{\sigma}\}$ resp. $\{ F_{j} \}$ we started with and they appeared in the $LT_{\mathbb{Z}}(R)$- module $\mathcal{O}_{sh}$ resp. the $LT_{\mathbb{Z}}(R_{as})$- module $\mathcal{O}_{as}$ as the formal exponential factor. We will now both give them a convergent footing. Note that for all $\sigma \in \Sigma$ and all $j \in \mathcal{J}$, the operator norms of the $M_{G_{\sigma}}$ and the $M_{F_{j}}$ satisfy
$$
||M_{G_{\sigma}}|| \leqslant 2 \text{ and }||M_{F_{j}}|| \leqslant 2.
$$
Therefore we choose our parameters $t_{\Sigma}=(t_{\sigma})$ and $t_{\mathcal{J}}=(t_{j}$ in respectively the spaces  $\ell_{1}(\Sigma)$ and $\ell_{1}(\mathcal{J})$ defined by 
$$
\ell_{1}(\Sigma)=\{ t_{\Sigma}=(t_{\sigma}) \mid \text{ all }t_{\sigma} \in \mathbb{R} \text{ and } \sum_{\sigma \in \Sigma} |t_{\sigma}| < \infty \}
$$
and 
$$
\ell_{1}(\mathcal{J})=\{ t_{\mathcal{J}}=(t_{j}) \mid \text{ all }t_{j} \in \mathbb{R} \text{ and } \sum_{j \in \mathcal{J}} |t_{j}| < \infty \}.
$$
We equip these two spaces respectively with the norms 
$$
||t_{\Sigma}||_{1}=\sum_{\sigma \in \Sigma} |t_{\sigma}| \text{ and } ||t_{\mathcal{J}}||_{1}=\sum_{j \in \mathcal{J}} |t_{j}|.
$$ 
Now we define two analytic maps. The first is $t_{\Sigma} \to \gamma_{\Sigma}(t_{\Sigma})$, with 
$$
\gamma_{\Sigma}(t_{\Sigma})=\exp( \sum_{\sigma \in \Sigma}  t_{\sigma} M_{G_{\sigma} })
$$
and maps $\ell_{1}(\Sigma)$ to $\GL(\mathcal{H}(\mathbb{C}))$. The second map is $t_{\mathcal{J}} \to \gamma_{\Sigma}(t_{\Sigma})$, where
$$
\gamma_{\mathcal{J}}(t_{\mathcal{J}})=\exp( \sum_{j \in \mathcal{J}}  t_{j} M_{F_{j} })
$$
and maps $\ell_{1}(\mathcal{J})$ to $\GL(\mathcal{H}(\mathbb{R}))$. 
The images of the maps $\gamma_{\Sigma}$ and $\gamma_{\mathcal{J}}$ we denote respectively by $\Gamma_{\Sigma}$ and $\Gamma_{\mathcal{J}}$.
Now we specify for both hierarchies the settings we work in. For the $\mathcal{C}_{sh}(\mathbb{C})$-hierarchy
we choose for the algebra $R(\mathbb{R})$ the $C^{\infty}$-functions on the space $\ell_{1}(\Sigma)$ with values in $\mathbb{R}$. Then $R$ consists of all $\mathbb{C}$-valued $C^{\infty}$-functions on $\ell_{1}(\Sigma)$, $[\gamma_{\Sigma}(t_{\Sigma})]$ is well-defined and all the coefficients of 
$$
[\gamma_{\Sigma}(t_{\Sigma})]=\exp( \sum_{\sigma \in \Sigma} t_{\sigma} G_{\sigma})
$$
belong to $R$. For each derivation $\partial_{\sigma}, \sigma \in \Sigma,$ we choose $\partial_{\sigma}=\frac{\partial}{\partial t_{\sigma}}$. In the case of the $\mathcal{C}_{as}(\mathbb{R})$-hierarchy
we choose for the algebra $R_{as}$ the $\mathbb{R}$-valued $C^{\infty}$-functions on the space $\ell_{1}(\mathcal{J})$. 
The matrix $[\gamma_{\mathcal{J}}(t_{\mathcal{J}})]$ is now well-defined and all the coefficients of 
$$
[\gamma_{\mathcal{J}}(t_{\mathcal{J}})]=\exp( \sum_{j \in \mathcal{J}} t_{j} F_{j})
$$
belong to $R_{as}$. For each derivation $\partial_{j}, j \in \mathcal{J},$ we choose $\partial_{j}=\frac{\partial}{\partial t_{j}}$.

The next step is the description of the group of transformations of the $\mathcal{H}(k)$, whose homogeneous space leads to solutions of the $\mathcal{C}_{sh}(\mathbb{C})$-hierarchy and the $\mathcal{C}_{as}(\mathbb{R})$-hierarchy. 
Our choice of the class of groups is similar to the ones that worked for the KP hierarchy \cite{Segal-Wilson}, its strict version \cite{HP3} and the discrete KP hierarchy and its strict version \cite{HPP2} and makes use of the two-sided ideal $S_{2}(\mathcal{H}(k))$ of Hilbert-Schmidt operators on $\mathcal{H}(k)$.
Recall it is defined by
$$
S_{2}(\mathcal{H}(k))=\{ A \in B(\mathcal{H}(k))\mid [A] =(a_{ij}), \sum_{i \in \mathbb{Z}}\sum_{j \in \mathbb{Z}} |a_{ij}|^{2} < \infty \}.
$$
By taking the transpose or the adjoint $S_{2}(\mathcal{H}(k))$ is clearly mapped bijectively onto itself  and becomes a Hilbert space w.r.t. the inner product
$$
(A,B)_{2}:=\sum_{n \in \mathbb{Z}}< \vec{e}\,(n) \mid A^{*}B \vec{e}\,(n) >.
$$
From the fact that the Hilbert-Schmidt operators form a two-sided ideal of compact operators, follows that one can introduce the group $G(2)(k) $ by
$$
G(2)(k)=\left\{ g 
\in \GL(\mathcal{H}(k)) \Biggm| g-\Id \in S_{2}(\mathcal{H}(k))
\right
\}.
$$
It is a normal subgroup of $\GL(\mathcal{H}(k))$, since $S_{2}(\mathcal{H}(k))$ is a two-sided ideal of $B(\mathcal{H}(k))$
One verifies directly that the space $S_{2}(\mathcal{H}(k))$ can be identified with the Lie algebra $\mathcal{G}(2)(k)$ of $G(2)(k)$ and that makes $G(2)(k)$ a Hilbert Lie group. Each 
$S_{2}(\mathcal{H}(k))$, $k=\mathbb{R} \text{ or }\mathbb{C} $, can be split into the direct sum of two Lie subalgebras corresponding to the decompositions of $LT_{\mathbb{Z}}(R)$ and $LT_{\mathbb{Z}}(R_{as})$ as treated in Section \ref{sec2}. For $k=\mathbb{C} $, the two subalgebras are
\begin{align}
\notag
&S_{2}(\mathcal{H}(\mathbb{C}))_{sh}=\{ A  \mid A^{*}=-A \} \text{ and }\\ \label{sh+-}
&S_{2}(\mathcal{H}(\mathbb{C}))_{-}=\{ A  \mid [A]=\sum_{i\leqslant 0} d_{i}\Lambda^{i}, d_{0} \in \mathcal{D}_{1}(\mathbb{R}), d_{i} \in \mathcal{D}_{1}(\mathbb{C}) \text{ for } i<0\}.
\end{align}
and for $k=\mathbb{R} $, the two subalgebras are
\begin{align}
\notag
&S_{2}(\mathcal{H}(\mathbb{R}))_{as}=\{ A  \mid A^{T}=-A \} \text{ and }\\ \label{as+-}
&S_{2}(\mathcal{H}(\mathbb{R}))_{\leqslant 0}=\{ A  \mid [A]=\sum_{i\leqslant 0} d_{i}\Lambda^{i}, \text{ for all }i \leqslant 0, d_{i} \in \mathcal{D}_{1}(\mathbb{R}) \}.
\end{align}
To the first of the two Lie subalgebras in (\ref{sh+-}) corresponds the subgroup $U(2)$ of all unitary operators in $G(2)(\mathbb{C})$ and to the second the group $P_{-}$ of all lower triangular matrices in $G(2)(\mathbb{C})$ that have on the central diagonal only real strict positive entries. Since $U(2) \cap P_{-}=\Id$, the map $I_{\mathbb{C}}: P_{-} \times U(2) \to G(2)(\mathbb{C})$ defined by $I_{\mathbb{C}}(p_{-},u))=p_{-}u$ is injective and locally a diffeomorphism. Beltita has shown in \cite{Beltita} that this map is also surjective. The decomposition of $G(2)(\mathbb{C})$ determined by the inverse of $I_{\mathbb{C}}$ is called its {\it Iwasawa decomposition}. The result of Beltita also implies that $\tilde{I}_{\mathbb{C}}(p_{-},u))=p_{-}^{-1}u$ is a diffeomorphism between between $P_{-} \times U(2)$ and $G(2)(\mathbb{C})$. We will use the twisted Iwasawa decompostion determined
by the inverse of $\tilde{I}_{\mathbb{C}}$: for each $g \in G(2)(\mathbb{C})$ there exist a $p_{-}(g) \in P_{-}$ and a $u(g) \in U(2)$ so that $g=p_{-}(g)^{-1}u(g)$.

Also both Lie subalgebras in (\ref{as+-}) correspond to Lie subgroups of $G(2)(\mathbb{R})$. The first is the Lie algebra of the orthogonal transformations $O(2)$ in $G(2)(\mathbb{R})$ and the second is the Lie algebra of the transformations $P_{\leqslant 0}$ in $G(2)(\mathbb{R})$ possessing a lower triangular matrix with a strict positive central diagonal. Also in this case we have $O(2) \cap P_{\leqslant 0}=\Id$ and thus the map $I_{\mathbb{R}}: P_{\leqslant 0} \times O(2) \to G(2)(\mathbb{R})$ defined by $I_{\mathbb{R}}(p_{\leqslant 0},o))=p_{-}o$ is injective and locally a diffeomorphism. In the same paper Beltita proved the surjectivity of $I_{\mathbb{R}}$ and the inverse of $I_{\mathbb{R}}$ is the Iwasawa decomposition of $G(2)(\mathbb{R})$. As in the complex case this implies that $\tilde{I}_{\mathbb{R}}(p_{\leqslant 0},o))=p_{-}^{-1}o$ is a diffeomorphism between the same varieties and we use the twisted Iwasasawa decomposition determined by the inverse of $\tilde{I}_{\mathbb{R}}$: for each $g \in G(2)(\mathbb{R})$ there exist a $p_{\leqslant 0}(g) \in P_{\leqslant 0}$ and a $o(g) \in O(2)$ so that $g=p_{\leqslant 0}(g)^{-1}o(g)$.

Now we come to the construction of the wave functions for both hierarchies. We start with the skew hermitian case. 
Take a $g \in G(2)(\mathbb{C})$, then conjugation of $g$ with an operator $\gamma_{\Sigma}(t_{\Sigma}) \in \Gamma_{\Sigma}$ yields a new element in $G(2)(\mathbb{C})$ to which we apply the twisted Iwasawa decomposition. Then we have for all $t_{\Sigma} \in \ell_{1}(\Sigma)$
$$
\gamma_{\Sigma}(t_{\Sigma}) g\gamma_{\Sigma}(t_{\Sigma})^{-1} =p_{-}(g)(t_{\Sigma})^{-1}u(g)(t_{\Sigma})
$$
and both $p_{-}(g)(t_{\Sigma})$ and $u(g)(t_{\Sigma})$ depend in a $C^{\infty}-$fashion of $t_{\Sigma}$. On the matrix level this identity yields
\begin{equation}
\label{wavedeco}
\Phi:=[p_{-}(g)(t_{\Sigma})][\gamma_{\Sigma}(t_{\Sigma})]=[p_{-}(g)(t_{\Sigma})]\varphi_{0}=[u(g)(t_{\Sigma})] [\gamma_{\Sigma}(t_{\Sigma})] [g]^{-1}.
\end{equation}
This shows that $\Phi$ is a candidate wave matrix for the $\mathcal{C}_{sh}(\mathbb{C})$-hierarchy, where all the products in $\Phi$ are no longer formal but real. Thanks to Proposition \ref{Pweak} we know that it suffices to show for all $\sigma \in \Sigma$ that
$$
\partial_{\sigma}(\Phi)=S_{\sigma} \Phi, \text{ with } S_{\sigma } \in \mathcal{SH}(R).
$$
On one hand we have that 
\begin{align*}
\partial_{\sigma}(\Phi)&=\{ \partial_{\sigma}([p(g)(t_{\Sigma})])[p(g)(t_{\Sigma})^{-1}] 
+[p(g)(t_{\Sigma})](G_{\sigma})[p(g)(t_{\Sigma})^{-1}]\}\Phi \\ \notag
&=S_{\sigma} \Phi, \text{ with } S_{\sigma} \in LT_{\mathbb{Z}}(R).
\end{align*}
On the other hand, if one applies $\partial_{\sigma}$ to the right hand side of equation (\ref{wavedeco}), then we get
\begin{align}
\label{rightfi}
\partial_{\sigma}(\Phi)&=\{ \partial_{\sigma}([u(g)(t_{\Sigma})])[u(g)(t_{\Sigma})^{-1}]  -[u(g)(t)]G_{\sigma}[u(g)(t_{\Sigma})]^{-1}\}\Phi\\ \notag
&=\{ \partial_{\sigma}([u(g)(t_{\Sigma})])[u(g)(t_{\Sigma})]^{*} -[u(g)(t_{\Sigma})]G_{\sigma}[u(g)(t_{\Sigma})]^{-1}\}\Phi
\end{align}
Since the matrix  $[u(g)(t_{\Sigma})] $ is unitary, applying $\partial_{\sigma}$ to the identity 
$$[u(g)(t_{\Sigma})] [u(g)(t_{\Sigma})]^{*}=\Id$$
shows that the matrix $\partial_{\sigma}([u(g)(t_{\Sigma})])[u(g)(t_{\Sigma})]^{*}$ is skew hermitian. 
Further, conjugating an skew hermitian matrix with a unitary one, yields again a skew hermitian matrix. Hence, also $[u(g)(t_{\Sigma})]G_{\sigma}[u(g)(t_{\Sigma})]^{-1}$ is skew hermitian and thus the right hand side of equation (\ref{rightfi}) is the product of a skew hermitian matrix and $\Phi$.
So, all the $S_{\sigma}$ are in $\mathcal{SH}(R)$ and $\Phi $ is a wave matrix for the $\mathcal{C}_{sh}(\mathbb{C})$-hierarchy. If we apply the construction to the element $gu$, with $u \in U(2)$, then we get
\begin{align*}
\gamma_{\Sigma}(t_{\Sigma}) gu \gamma_{\Sigma}(t_{\Sigma})^{-1} &=\gamma_{\Sigma}(t_{\Sigma}) g\gamma_{\Sigma}(t_{\Sigma})^{-1} \gamma_{\Sigma}(t_{\Sigma}) u \gamma_{\Sigma}(t_{\Sigma})^{-1}\\
&= p(g)(t_{\Sigma})^{-1} u(g)(t_{\Sigma}) \gamma_{\Sigma}(t_{\Sigma}) u \gamma_{\Sigma}(t_{\Sigma})^{-1}.
\end{align*}
Now each $\gamma_{\Sigma}(t_{\Sigma})$ is a unitary matrix in $\GL(\mathcal{H}(\mathbb{C}))$ and 
$S_{2}(\mathcal{H}(\mathbb{C}))$ is a two sided ideal in $B(\mathcal{H}(\mathbb{C}))$, so conjugating $u$ with $\gamma_{\Sigma}(t_{\Sigma})$ results in an element of $U(2)$. Hence we may conclude $p(gu)(t_{\Sigma})=p(g)(t_{\Sigma})$ and $u(gu)(t_{\Sigma})=u(g)(t_{\Sigma}) \gamma_{\Sigma}(t_{\Sigma}) u \gamma_{\Sigma}(t_{\Sigma})^{-1}$. So, both $g$ and $gu$ generate the same solution of the $\mathcal{C}_{sh}(\mathbb{C})$-hierarchy. 

The construction of wave matrices for the $\mathcal{C}_{as}(\mathbb{R})$-hierarchy starts with an element 
$g \in G(2)(\mathbb{R})$ and follows the same line of approach as in the skew hermitian case. We conjugate $g$ with an operator $\gamma_{\mathcal{J}}(t_{\mathcal{J}})$ from $\Gamma_{\mathcal{J}}$ in $G(2)(\mathbb{R})$ to which we apply the twisted Iwasawa decomposition of this group.
 Then we have for all $t_{\mathcal{J}} \in \ell_{1}(\mathcal{J})$
$$
\gamma_{\mathcal{J}}(t_{\mathcal{J}}) g\gamma_{\mathcal{J}}(t_{\mathcal{J}})^{-1} =p_{\leqslant 0}(g)(t_{\mathcal{J}})^{-1}o(g)(t_{\mathcal{J}})
$$
and both $p_{\leqslant 0}(g)(t_{\mathcal{J}})$ and $o(g)(t_{\mathcal{J}})$ depend in a $C^{\infty}-$fashion of $t_{\mathcal{J}}$. On the matrix level this identity yields
\begin{equation}
\label{wavedeco2}
\Psi:=[p_{\leqslant 0}(g)(t_{\mathcal{J}})][\gamma_{\mathcal{J}}(t_{\mathcal{J}})]=[p_{\leqslant 0(g)}(t_{\mathcal{J}})]\psi_{0}=[o(g)(t_{\mathcal{J}})] [\gamma_{\mathcal{J}}(t_{\mathcal{J}})] [g]^{-1}.
\end{equation}
This shows that $\Psi$ is a candidate wave matrix for the $\mathcal{C}_{as}(\mathbb{R})$-hierarchy, where all the products in both expressions for $\Psi$ are real. 
Thanks to Proposition \ref{Pweak} we know that it suffices to show for all $j \in \mathcal{J}$ that
$$
\partial_{j}(\Psi)=A_{j} \Psi, \text{ with } A_{j } \in \mathcal{AS}(R_{as}).
$$
On one hand we have that 
\begin{align*}
\partial_{j}(\Psi)&=\{ \partial_{j}([p_{\leqslant 0}(g)(t_{\mathcal{J}})])[p_{\leqslant 0}(g)(t_{\mathcal{J}})] ^{-1}
+[p_{\leqslant 0}(g)(t_{\mathcal{J}})]F_{j}[p_{\leqslant 0}(g)(t_{\mathcal{J}})]^{-1}\}\Psi \\ \notag
&=A_{j} \Psi, \text{ with } A_{j} \in LT_{\mathbb{Z}}(R).
\end{align*}
On the other hand, if one applies $\partial_{j}$ to the right hand side of equation (\ref{wavedeco2}), then we get
\begin{align}
\label{rightpsi}
\partial_{j}(\Psi)&=\{ \partial_{j}([o(g)(t_{\mathcal{J}})])[o(g)(t_{\mathcal{J}})^{-1}]  -[o(g)(t_{\mathcal{J}})]F_{j}[o(g)(t_{\mathcal{J}})]^{-1}\}\Psi\\ \notag
&=\{ \partial_{j}([o(g)(t_{\mathcal{J}})])[o(g)(t_{\mathcal{J}})]^{T} -[o(g)(t_{\mathcal{J}})]F_{j}[o(g)(t_{\mathcal{J}})]^{-1}\}\Psi
\end{align}
Since the matrix  $[o(g)(t)] $ is orthogonal, applying $\partial_{j}$ to the identity 
$$[o(g)(t_{\mathcal{J}})] [o(g)(t_{\mathcal{J}})]^{T}=\Id$$
shows that the matrix $\partial_{j}([o(g)(t_{\mathcal{J}})])[o(g)(t_{\mathcal{J}})]^{T}$ is anti-symmetric. 
Further, conjugating an anti-symmetric matrix with an orthogonal one, yields again an anti-symmetric matrix. Hence, also $[o(g)(t_{\mathcal{J}})]F_{j}[o(g)(t_{\mathcal{J}})]^{-1}$ is anti-symmetric and thus the right hand side of equation (\ref{rightpsi}) is the product of an anti-symmetric matrix and $\Psi$.
So, all the $A_{j}$ are in $\mathcal{AS}(R_{as})$ and $\Psi $ is a wave matrix for the $\mathcal{C}_{as}(\mathbb{R})$-hierarchy. If we apply the construction to the element $go$, with $o \in O(2)$, then we get
\begin{align*}
\gamma_{\mathcal{J}}(t_{\mathcal{J}}) go \gamma_{\mathcal{J}}(t_{\mathcal{J}})^{-1} &=\gamma_{\mathcal{J}}(t_{\mathcal{J}}) g\gamma_{\mathcal{J}}(t_{\mathcal{J}})^{-1} \gamma_{\mathcal{J}}(t_{\mathcal{J}}) 0 \gamma_{\mathcal{J}}(t_{\mathcal{J}})^{-1} \\
&= p(g)(t_{\mathcal{J}}) o(g)(t_{\mathcal{J}}) \gamma_{\mathcal{J}}(t_{\mathcal{J}}) o \gamma_{\mathcal{J}}(t_{\mathcal{J}})^{-1}.
\end{align*}
Since $\gamma_{\mathcal{J}}(t_{\mathcal{J}})$ is an orthogonal element of $\GL(\mathcal{H}(\mathbb{R}))$ and 
$S_{2}(\mathcal{H}(\mathbb{R}))$ is a two sided ideal in $B(\mathcal{H}(\mathbb{R}))$, conjugation of $o$ with 
$\gamma_{\mathcal{J}}(t_{\mathcal{J}})$ yields an element of $O(2)$.
Hence there follows $p_{\leqslant 0}(go)(t_{\mathcal{J}})=p_{\leqslant 0}(g)(t_{\mathcal{J}})$ and $o(go)(t_{\mathcal{J}})=o(g)(t_{\mathcal{J}}) \gamma_{\mathcal{J}}(t_{\mathcal{J}}) o \gamma_{\mathcal{J}}(t_{\mathcal{J}})^{-1}$. So, the elements $g$ and $go$ generate the same solution of the $\mathcal{C}_{as}(\mathbb{R})$-hierarchy. 
We summarize the results in a
\begin{theorem}
\label{T1}
Solutions of the $\mathcal{C}_{sh}(\mathbb{C})$-hierarchy and the $\mathcal{C}_{as}(\mathbb{R})$-hierarchy
can be obtained as follows:
\begin{enumerate}
\item[(a)] 
For each $g \in G(2)(\mathbb{C})$ and each $\sigma \in \Sigma$, consider the perturbation $$\mathcal{G}_{\sigma}(g):=[p_{-}(g)(t_{\Sigma})]G_{\sigma}[p_{-}(g)(t_{\Sigma})]^{-1}$$ of the basic direction $G_{\sigma}$, with the $\mathbb{Z} \times \mathbb{Z}$-matrix $[p_{-}(g)(t_{\Sigma})]$ from (\ref{wavedeco}). Then the 
$\{ \mathcal{G}_{\sigma}(g) \}$ are a solution of the $\mathcal{C}_{sh}(\mathbb{C})$-hierarchy. This solution does not change if one replaces $g$ by $gu$, with $u \in U(2)$. Hence the symmetric space $G(2)(\mathbb{C})/U(2)$ determines the solutions of the $\mathcal{C}_{sh}(\mathbb{C})$-hierarchy.
\item[(b)]
For each $g \in G(2)(\mathbb{R})$ and each $j \in \mathcal{J}$, consider the perturbation $$\mathcal{F}_{j}(g):=[p_{\leqslant 0}(g)(t_{\mathcal{J}})]F_{j}[p_{\leqslant 0}(g)(t_{\mathcal{J}})]^{-1}$$ of the basic direction $F_{j} $, with the $\mathbb{Z} \times \mathbb{Z}$-matrix $[p_{\leqslant 0}(g)(t_{\mathcal{J}})]$ from (\ref{wavedeco2}). Then the 
$\{ \mathcal{F}_{j}(g) \}$ are a solution of the $\mathcal{C}_{as}(\mathbb{R})$-hierarchy. This solution does not change if one replaces $g$ by $go$, with $o \in O(2)$. Hence the symmetric space $G(2)(\mathbb{R})/O(2)$ determines the solutions of the $\mathcal{C}_{sh}(\mathbb{R})$-hierarchy.
\end{enumerate}
\end{theorem}

\end{document}